\documentclass[12pt]{article}
\usepackage{amsmath}
\usepackage{amssymb}
\usepackage{prodint}
\usepackage{amsfonts}
\usepackage{graphicx,psfrag,epsf}
\usepackage{natbib}
\usepackage[utf8]{inputenc}

\usepackage{url} 
\usepackage{appendix}

\usepackage{amsthm}

\usepackage[figuresright]{rotating}

\def\bSig\mathbf{\Sigma}

\newtheorem{theorem}{Theorem}
\newtheorem{lemma}{Lemma}

\begin{document}


  \title{\bf Nonparametric Analysis of Nonhomogeneous Multi-State Processes based on Clustered Observations}
  \author{Giorgos Bakoyannis \\
    Department of Biostatistics, Indiana University}
  \maketitle

\def\spacingset#1{\renewcommand{\baselinestretch}%
{#1}\small\normalsize} \spacingset{1}


\begin{abstract}
Frequently, clinical trials and observational studies involve complex event history data with multiple events. When the observations are independent, the analysis of such studies can be based on standard methods for multi-state models. However, the independence assumption is often violated, such as in multicenter studies, which makes the use of standard methods improper. In this work we address the issue of nonparametric estimation and two-sample testing for the population-averaged transition and state occupation probabilities under general multi-state models based on right-censored, left-truncated, and clustered observations. The proposed methods do not impose assumptions regarding the within-cluster dependence, allow for informative cluster size, and are applicable to both Markov and non-Markov processes. Using empirical process theory, the estimators are shown to be uniformly consistent and to converge weakly to tight Gaussian processes. Closed-form variance estimators are derived, rigorous methodology for the calculation of simultaneous confidence bands is proposed, and the asymptotic properties of the nonparametric tests are established. Furthermore, we provide theoretical arguments for the validity of the nonparametric cluster bootstrap, which can be readily implemented in practice regardless of how complex the underlying multi-state model is. Simulation studies show that the performance of the proposed methods is good, and that methods that ignore the within-cluster dependence can lead to invalid inferences. Finally, the methods are applied to data from a multicenter randomized controlled trial.
\end{abstract}

\noindent%
\vfill

\newpage
\spacingset{1.45} 
\section{Introduction}
\label{s:intro}
Frequently, clinical trials and observational studies involve complex multi-state event histories. An example is cancer clinical trials where patient event histories typically involve three or more clinical states, such as ``cancer-free'', ``cancer'', and ``death''. Another example is observational studies of HIV-infected individuals in HIV care programs. In such studies, infected individuals can start antiretroviral treatment (ART), have a gap in care, return to care after a gap, and die after being in any of the aforementioned states. 
When the observations are independent, nonparametric estimation of the transition probabilities for such multi-state processes can be performed using the Aalen--Johansen estimator \citep{Aalen78}. Calculation of simultaneous confidence bands and nonparametric two-sample hypothesis tests can be performed using the recently proposed approaches by \citet{Bluhmki18} and \citet{Bakoyannis19b}, respectively.

The independent observations assumption is often violated in medical research. This is typical in multicenter studies, where the events of individuals within the same center are expected to be associated. Such a multicenter study is the motivating European Organization for Research and Treatment of Cancer (EORTC) trial 10854, which evaluated the effectiveness of the combination of surgery with polychemotherapy compared to surgery alone as a treatment for early breast cancer. In total, 2793 early breast cancer patients from 15 hospitals (i.e. centers/clusters) were recruited in this trial. The patient event history in this trial involved three states: i) cancer-free, ii) cancer relapse, and iii) death. When the observations exhibit within-cluster dependence, the traditional Greenwood standard error estimators for the transition probabilities, the simultaneous confidence bands by \citet{Bluhmki18}, and the nonparametric tests by \citet{Bakoyannis19b} are not valid.

Several parametric methods have been proposed for the analysis of multi-state models based on clustered observations \citep{Cook04,Li15,Yiu18}. However, these methods impose strong parametric assumptions about the underlying multi-state processes which are expected to be violated in practice. \citet{Chen13} proposed a semiparametric random-effects approach for cluster-specific inference about non-homogeneous Markov processes. This approach, which also allows for non-ignorable missingness, utilizes a Monte-Carlo EM algorithm. Recently, \citet{OKeeffe18} proposed a nonparametric approach for cluster-specific inference based on correlated observations from a general multi-state model. This approach, similarly to the \citet{Chen13} method, accounts for the within-cluster dependence by incorporating random effects. Estimation in this case relies on numerical integration. There are no other nonparametric approaches for clustered multi-state data that utilize random effects that we are aware of. The current semiparametric and nonparametric proposals for clustered observations that utilize random effects \citep{Chen13,OKeeffe18} have several limitations. First, they impose strong parametric assumptions on the random effects. Also, these random effects introduce only a restrictive positive within-cluster association. Second, they tend to be computationally intensive which may restrict their use with larger data sets. Third, they do not establish the asymptotic properties of the proposed estimators for the transition probabilities. Moreover, they do not provide methodology for simultaneous confidence bands and nonparametric hypothesis testing. Fourth, they do not consider the case of informative cluster size, where there is an association between cluster size and observed events. Finally, in many applications, population-averaged inference is more scientifically relevant than cluster-specific inference. This is the case with our motivating EORTC trial 10854. To our knowledge, only \citet{Dipankar17} proposed a method for nonparametric population-averaged inference about state occupation probabilities in general multi-state models. Importantly, \citet{Dipankar17} allow for informative cluster size. However, this approach is for current status data and not the usual right-censored or left-truncated multi-state data. Moreover, the asymptotic properties of this method have not been established, and there is no methodology for simultaneous confidence bands and nonparametric tests.

To the best of our knowledge, the issue of nonparametric population-averaged inference for event probabilities in general multi-state models based on right-censored, left-truncated, and clustered observations has not been addressed so far. In this work, we address this issue by proposing rigorous estimators and methodology for standard error estimation, simultaneous confidence bands, and nonparametric two-sample Kolmogorov--Smirnov-type tests. The asymptotic properties of the proposed methods are rigorously established using modern empirical process theory and closed-form variance estimators are provided. In addition, we establish the validity of the nonparametric cluster bootstrap and show how it can be used for the calculation of simultaneous confidence bands and $p$-values. This is particularly useful in practice, since it provides a convenient way to conduct inference using off-the-shelf software. The proposed methods do not impose restrictive parametric assumptions or assumptions regarding the within-cluster dependence. We additionally allow for informative cluster size and nonhomogeneous processes which are non-Markov. Simulation studies show that the methods perform well and that standard methods for independent observations provide severely under-estimated standard errors and confidence bands with a poor coverage rate. Finally, the methods are applied to the data from the multicenter EORTC trial 10854.

\section{Nonparametric estimation}
\label{s:method}
\subsection{Independent observations}
Consider a continuous time nonhomogeneous Markov process $\{X(t):t\in [0,\tau]\}$, for some $\tau<\infty$, with a finite state space $\mathcal{S}=\{1,\ldots,k\}$ and a subspace $\mathcal{T}\subset\mathcal{S}$ that includes the possible absorbing states (e.g. death). For situations without absorbing states we set $\mathcal{T}=\varnothing$. The Markov assumption will be relaxed later in subsection \ref{ss:nonmark}. The stochastic behavior of the process can be described by the $k\times k$ transition probability matrix $\tilde{\mathbf{P}}_{0}(s,t)$, with elements
\begin{eqnarray*}
\tilde{P}_{0,hj}(s,t)&=&\Pr(X(t)=j|X(s)=h,\mathcal{F}_{s^-}) \nonumber \\
&=&\Pr(X(t)=j|X(s)=h) \ \ \ \ h,j\in\mathcal{S}, \ \ 0\leq s< t\leq \tau,
\end{eqnarray*}
where $\mathcal{F}_{s^-}=\sigma\big\langle\{\check{N}_{hj}(u):0\leq u< s,h\neq j\}\big\rangle$ is the event history prior to time $s$, with $\check{N}_{hj}(t)$ being the number of direct transitions from state $h$ to state $j$, $h\neq j$, in $[0,t]$ in the absense of right censoring or left truncation. Note that the conditional independence from the prior history $\mathcal{F}_{s^-}$ above is the Markov assumption. If the transition probabilities are absolutely continuous then the transition intensities are defined as $\tilde{\alpha}_{0,hj}(t)=\lim_{\delta\downarrow 0}\tilde{P}_{0,hj}(t,t+\delta)/\delta$ for $h\in\mathcal{T}^c$ and $j\in\mathcal{S}$, where $\tilde{a}_{0,hh}(t)=-\sum_{j\neq h} \tilde{a}_{0,hj}(t)$. Another key quantity is the cumulative transition intensity which is defined as $\tilde{A}_{0,hj}(t)=\int_0^t \tilde{a}_{0,hj}(u)du$ for the absolute continuous case, or more generally, by the Kolmogorov forward equation \citep{Aalen08}, as $\tilde{A}_{0,hj}(t)=\int_0^t[E\check{Y}_h(u)]^{-1}dE\check{N}_{hj}(u)$, $h\neq j$, $t\in[0,\tau]$, with $\tilde{A}_{0,hh}(t)=-\sum_{j\neq h} \tilde{A}_{0,hj}(t)$, where $\check{Y}_h(t)$ is the at-risk process for state $h$, with $\check{Y}_h(t)=1$ if the process is at state $h$ just before time $t$ and $Y_h(t)=0$ otherwise. Based on the $k\times k$ matrix $\tilde{\mathbf{A}}_{0}(t)$ of cumulative transition intensities, the transition probability matrix can be defined as $\tilde{\mathbf{P}}_{0}(s,t)=\prodi_{(s,t]}[\mathbf{I}_k+d\tilde{\mathbf{A}}_{0}(u)]$, $t\in[0,\tau]$ where $\prodi$ is the product integral and $\mathbf{I}_k$ is the $k\times k$ identity matrix \citep{Andersen12}. Another quantity of interest is the state occupation probability $\tilde{P}_{0,j}(t)=\Pr(X(t)=j)$, which is defined as $\tilde{P}_{0,j}(t)=\sum_{h\in\mathcal{T}^c}\tilde{P}_{0,h}(0)\tilde{P}_{0,hj}(0,t)$, $j\in\mathcal{S}$, $t\in[0,\tau]$ \citep{Andersen12}. Estimation of the cumulative transition intensities based on independent observations of the process $X(\cdot)$ can be achieved using the nonparametric Nelson--Aalen estimator, which allows for both independent right censoring and left truncation. The Nelson--Aalen estimator can be used to obtain natural plug-in estimators of the transition probability matrix and the state occupation probabilities \citep{Andersen12}. For the latter case, the estimator of $\tilde{P}_{0,h}(0)$ is the sample proportion of the observations at state $h$ at time $t=0$. 

\subsection{Clustered observations}
\label{ss:clust}
Suppose that a study involves $n$ clusters of observations of the Markov process $\{X(t):t\in[0,\tau]\}$, with $M_i$ observations in the $i$th cluster. The observable data are the possibly right-censored and/or left-truncated counting processes $\{N_{im,hj}(t):h\neq j,t\in[0,\tau]\}$ and the at-risk processes $\{Y_{im,h}(t):h\in\mathcal{T}^c,t\in[0,\tau]\}$, for $i=1,\ldots,n$ and $m=1,\ldots,M_i$. Note that $\{N_{im,hj}(t):h\neq j,t\in[0,\tau]\}$ and $\{Y_{im,h}(t):h\in\mathcal{T}^c,t\in[0,\tau]\}$ are the observable versions of the complete (i.e. non-right-censored and non-left-truncated) processes $\{\check{N}_{im,hj}(t):h\neq j,t\in[0,\tau]\}$ and $\{\check{Y}_{im,h}(t):h\in\mathcal{T}^c,t\in[0,\tau]\}$. For a randomly selected cluster member $m_i$, the stochastic processes $\{N_{im_i,hj}(t):h\neq j,t\in[0,\tau]\}$ and $\{Y_{im_i,h}(t):h\in\mathcal{T}^c,t\in[0,\tau]\}$ for $i=1,\ldots,n$ are assumed to be i.i.d.. However, the individual counting and at-risk processes are allowed to be dependent within clusters, with an arbitrary dependence structure. In this article we assume that the cluster sizes $M_i$, $i=1,\ldots,n$, are i.i.d. random positive integers. Furthermore, we allow the counting and at risk-processes to depend on cluster size $M_i$ (informative or nonignorable cluster size). However, the methods we present here are trivially applicable to simpler situations where cluster size $M_i$ is either non-informative or fixed.

In general, when the cluster size is random and informative, there are two population-averaged parameters of interest 
\citep{Seaman14}. The first one corresponds to the population of all cluster members, while the second concerns the population of typical cluster members. The population-averaged state occupation probabilities over the population of all cluster members are defined, similarly to marginal generalized linear models \citep{Seaman14}, as $P_{0,j}(t) = E[M_1I(X_{1m}(t)=j)]\big/EM_1$, $j\in\mathcal{S}$, $t\in[0,\tau]$, for a randomly selected cluster member $m$. This can be seen as a weighted average where larger clusters have a larger influence on the estimand. The population-averaged state occupation probabilities over the population of typical cluster members are defined as $P_{0,j}'(t) = EI(X_{1m}(t)=j)$, $j\in\mathcal{S}$, for a randomly selected cluster member $m$. In this case all clusters contribute a single (randomly selected) member and, therefore, all cluster have the same weight on the estimand. 
The two versions of the population-averaged transition probabilities can be defined similarly to $P_{0,j}(t)$ and $P_{0,j}'(t)$, for $j\in\mathcal{S}$. This leads to the population-averaged cumulative transition intensities $A_{0,hj}(t)=\int_0^t\{E[M_1\check{Y}_{1m,h}(u)]\}^{-1}dE[M_1\check{N}_{1m,hj}(u)]$, $h\neq j$, with $A_{0,hh}(t)=-\sum_{j\neq h}A_{0,hj}(t)$, and $A_{0,hj}'(t)=\int_0^t\{E[\check{Y}_{1m,h}(u)]\}^{-1}dE[\check{N}_{1m,hj}(u)]$, $h\neq j$, with $A_{0,hh}'(t)=-\sum_{j\neq h}A_{0,hj}'(t)$. Based on the corresponding population-averaged matrices $\mathbf{A}_{0}(t)$ and $\mathbf{A}_{0}'(t)$, the population-averaged transition probability matrices can be expressed as the product integrals (which are the solution to the Kolmogorov forward equations) $\mathbf{P}_{0}(s,t)=\prodi_{(s,t]}[\mathbf{I}_k+d\mathbf{A}_{0}(u)]$, $0\leq s\leq t \leq\tau$ and $\mathbf{P}_{0}'(s,t)=\prodi_{(s,t]}[\mathbf{I}_k+d\mathbf{A}_{0}'(u)]$, $0\leq s\leq t \leq\tau$. It is important to note that the most appropriate estimand depends on the scientific question of interest. In the special case where cluster size is either non-informative or constant $\mathbf{P}_0=\mathbf{P}_0'$ and $P_{0,j}=P_{0,j}'$, $j\in\mathcal{S}$. 

\subsection{Estimation of transition probabilities}
\label{ss:TP}
Consistent nonparametric estimation of the population-averaged transition probability matrices $\mathbf{P}_0$ and $\mathbf{P}_0'$ can be achieved by plugging consistent nonparametric estimators of the cumulative transition intensity matrices $\mathbf{A}_0$ and $\mathbf{A}_0'$ into the corresponding product integrals defined in \ref{ss:clust}. Let $N_{i\cdot,hj}(t)\equiv \sum_{m=1}^{M_i}N_{im,hj}(t)$ and $\check{N}_{i\cdot,hj}(t)\equiv \sum_{m=1}^{M_i}\check{N}_{im,hj}(t)$, for $h\neq j$. Similarly, let $Y_{i\cdot,h}(t)\equiv \sum_{m=1}^{M_i}Y_{im,h}(t)$ and $\check{Y}_{i\cdot,h}(t)\equiv \sum_{m=1}^{M_i}\check{Y}_{im,h}(t)$, for $h\in\mathcal{T}^c$. In Appendix A.1 we show that $E[\check{N}_{1\cdot,hj}(t)]=E[M_1\check{N}_{1m,hj}(t)]$, $h\neq j$, and $E[\check{Y}_{1\cdot,h}(t)]=E[\check{M}_1Y_{1m,h}(t)]$, $h\in\mathcal{T}^c$, $t\in[0,\tau]$, for any cluster member $m=1,\ldots,M_1$. This implies that $A_{0,hj}(t)=\int_0^t\{E[\check{Y}_{1\cdot,h}(u)]\}^{-1}dE[\check{N}_{1\cdot,hj}(u)]$, $h\neq j$. Furthermore, we show in Appendix A.2 that, under independent right censoring and left truncation, $A_{0,hj}(t)=\int_0^t\{E[Y_{1\cdot,h}(u)]\}^{-1}dE[N_{1\cdot,hj}(u)]$, $h\neq j$. Therefore, a natural estimator of $A_{0,hj}(t)$ is
\[
\hat{A}_{n,hj}(t)=\int_0^t\frac{d\left[\sum_{i=1}^n{N_{i\cdot,hj}(u)}\right]}{\sum_{i=1}^nY_{i\cdot,h}(u)}, \ \ \ \ h\neq j, \ \ t\in[0,\tau].
\]
Similar arguments lead to the conclusion that $A_{0,hj}'(t)=\int_0^t\{E[M_i^{-1}\check{Y}_{i\cdot,h}(u)]\}^{-1}dE[M_i^{-1}\check{N}_{i\cdot,hj}(u)]=\int_0^t\{E[M_i^{-1}Y_{i\cdot,h}(u)]\}^{-1}dE[M_i^{-1}N_{i\cdot,hj}(u)]$, $h\neq j$. Therefore, a natural nonparametric estimator of $A_{0,hj}'(t)$ is
\[
\hat{A}_{n,hj}'(t)=\int_0^t\frac{d\left[\sum_{i=1}^nM_i^{-1}N_{i\cdot,hj}(u)\right]}{\sum_{i=1}^nM_i^{-1}Y_{i\cdot,h}(u)}, \ \ \ \ h\neq j, \ \ t\in[0,\tau].
\]
Then, the proposed plug-in estimators of $\mathbf{P}_0$ and $\mathbf{P}_0'$ are
\[
\hat{\mathbf{P}}_n(s,t)=\Prodi_{(s,t]}[\mathbf{I}_k+d\hat{\mathbf{A}}_n(u)] \ \ \ \ \textrm{and} \ \ \ \ \hat{\mathbf{P}}_n'(s,t)=\Prodi_{(s,t]}[\mathbf{I}_k+d\hat{\mathbf{A}}_n'(u)],
\]
where $\hat{\mathbf{A}}_n(t)$ and $\hat{\mathbf{A}}_n'(t)$ are the $k\times k$ matrices with off-diagonal elements $\hat{A}_{n,hj}(t)$ and $\hat{A}_{n,hj}'(t)$, and diagonal elements $-\sum_{j\neq h}\hat{A}_{n,hj}(t)$ and $-\sum_{j\neq h}\hat{A}_{n,hj}'(t)$, $h=1,\ldots,k$, respectively. The estimator $\hat{\mathbf{P}}_n$ can be seen as the working independence Aalen--Johansen estimator. We call $\hat{\mathbf{P}}_n'$ the weighted by cluster size working independence Aalen--Johansen estimator. The following theorem states that the proposed estimators $\hat{\mathbf{P}}_n$ and $\hat{\mathbf{P}}_n'$ are uniformly consistent for the corresponding true population-averaged transition probability matrices $\mathbf{P}_0$ and $\mathbf{P}_0'$.

\begin{theorem}
Suppose that conditions C1--C5 in Appendix A.1 hold and define the norm $\|\mathbf{A}\|=\sup_{l}\sum_r|a_{lr}|$ for some matrix $\mathbf{A}=[a_{lr}]$. Then
\[
\sup_{t\in[s,\tau]}\left\|\hat{\mathbf{P}}_n(s,t)-\mathbf{P}_0(s,t)\right\|\overset{as*}\rightarrow 0 \ \ \ \ \textrm{and} \ \ \ \ \sup_{t\in[s,\tau]}\left\|\hat{\mathbf{P}}_n'(s,t)-\mathbf{P}_0'(s,t)\right\|\overset{as*}\rightarrow 0,
\]
for any $s\in[0,\tau]$, as $n\rightarrow\infty$.
\end{theorem}
The proof of Theorem 1 can be found in Appendix A.2. It has to be noted that, even though the standard Aalen--Johansen estimator under the working independence assumption is uniformly consistent for $\mathbf{P}_0$, the usual standard error estimators for the Aalen--Johansen estimator are invalid with clustered data as they ignore the within-cluster dependence. 

Theorem 2 provides the basis for valid inference about the components of $\mathbf{P}_0$ and $\mathbf{P}_0'$. Before stating Theorem 2 we define the functions
\[
\gamma_{ihj}(s,t)=\sum_{l\in\mathcal{T}^c}\sum_{q\in\mathcal{S}}\int_s^t\frac{P_{0,hl}(s,u-)P_{0,qj}(u,t)}{E\left[Y_{1\cdot,l}(u)\right]}d\bar{M}_{ilq}(u), \ \ \ \ 0\leq s\leq t \leq \tau,
\]
for $h\in\mathcal{T}^c$ and $j\in\mathcal{S}$, with $h\neq j$, where $\bar{M}_{ilq}(t)=N_{i\cdot,lq}(t)-\int_{(0,t]}Y_{i\cdot,l}(u)dA_{0,lq}(u)$
. If $h=j$, then $\gamma_{ihh}(s,t)=-\sum_{j\neq h}\gamma_{ihj}(s,t)$. Also, define
\[
\gamma_{ihj}'(s,t)=\sum_{l\in\mathcal{T}^c}\sum_{q\in\mathcal{S}}\int_s^t\frac{P_{0,hl}'(s,u-)P_{0,qj}'(u,t)}{E\left[M_1^{-1}Y_{1\cdot,l}(u)\right]}d\bar{M}_{ilq}'(u), \ \ \ \ 0\leq s\leq t \leq \tau,
\]
for $h\in\mathcal{T}^c$ and $j\in\mathcal{S}$, with $h\neq j$, where $\bar{M}_{ilq}'(t)=M_i^{-1}[N_{i\cdot,lq}(t)-\int_{(0,t]}Y_{i\cdot,l}(u)dA_{0,lq}'(u)]$. If $h=j$, then $\gamma_{ihh}'(s,t)=-\sum_{j\neq h}\gamma_{ihj}'(s,t)$. Next, define the estimated process $\hat{B}_{n,hj}(s,\cdot)=n^{-1/2}\sum_{i=1}^n\hat{\gamma}_{ihj}(s,\cdot)\xi_i$, for $h\in\mathcal{T}^c$ and $j\in\mathcal{S}$, where $\xi_i$, $i=1,\ldots,n$, are i.i.d. standard normal random variables, and $\hat{\gamma}_{ihj}(s,\cdot)$ is an estimated version of $\gamma_{ihj}(s,\cdot)$ where unknown quantities have been replaced by their uniformly consistent estimates and expectations by sample averages. Similarly, we define the estimated process $\hat{B}_{n,hj}'(s,\cdot)=n^{-1/2}\sum_{i=1}^n\hat{\gamma}_{ihj}'(s,\cdot)\xi_i$, for $h\in\mathcal{T}^c$ and $j\in\mathcal{S}$. 
These estimated processes will be shown useful for the calculation of simultaneous confidence bands and $p$-values for the two-sample comparison problem. These procedures utilize the notion weak convergence of conditional laws of the processes $\hat{B}_{n,hj}(s,\cdot)$ and $\hat{B}_{n,hj}'(s,\cdot)$ conditionally on the observed data \citep[see][]{Kosorok08}. Clearly, conditionally on the observed data, the only source of randomness are the standard normal variates $\xi_i$. 
Weak convergence of conditional laws is denoted as $\overset{p}{\underset{\xi}\leadsto}$.

An alternative approach for simultaneous confidence bands and calculation of $p$-values is the nonparametric cluster bootstrap. The nonparametric cluster bootstrap versions of the proposed estimators are $\hat{\mathbf{P}}_n^*(s,t)=\Prodi_{(s,t]}[\mathbf{I}_k+d\hat{\mathbf{A}}_n^*(u)]$ and $\hat{\mathbf{P}}_n^{\prime *}(s,t)=\Prodi_{(s,t]}[\mathbf{I}_k+d\hat{\mathbf{A}}_n^{\prime *}(u)]$, where $\hat{\mathbf{A}}_n^*(t)$ and $\hat{\mathbf{A}}_n^{\prime *}(t)$ involve the components
\[
\hat{A}_{n,hj}^*(t)=\int_0^t\frac{d\left[\sum_{i=1}^nU_{ni}{N_{i\cdot,hj}(u)}\right]}{\sum_{i=1}^nU_{ni}Y_{i\cdot,h}(u)}, \ \ \ \ h\neq j, \ \ t\in[0,\tau],
\]
and
\[
\hat{A}_{n,hj}^{\prime *}(t)=\int_0^t\frac{d\left[\sum_{i=1}^nU_{ni}M_i^{-1}{N_{i\cdot,hj}(u)}\right]}{\sum_{i=1}^nU_{ni}M_i^{-1}Y_{i\cdot,h}(u)}, \ \ \ \ h\neq j, \ \ t\in[0,\tau],
\]
respectively. $(U_{n1},\ldots,U_{nn})$ is a random vector from the multinomial distribution with $n$ trials and probabilities $1/n$ for each trial. Calculation of a bootstrap realization $\hat{\mathbf{P}}_n^*$ and $\hat{\mathbf{P}}_n^{\prime *}$ can be easily performed by randomly sampling $n$ clusters with replacement from the original data set, followed by the calculation of the proposed estimators based on the resulting bootstrap data set. Weak convergence of conditional laws of the nonparametric cluster bootstrap processes is defined, conditionally on the observed data, with respect to the multinomial bootstrap weights $U$ and is denoted as $\overset{p}{\underset{U}\leadsto}$.

\begin{theorem}
Suppose that conditions C1--C6 in Appendix A.1 hold. Then, for any $h\in\mathcal{T}^c$, $j\in\mathcal{S}$, and $s\in[0,\tau)$,
\begin{itemize}
\item[(i)] $
\sqrt{n}[\hat{P}_{n,hj}(s,t)-P_{0,hj}(s,t)]=n^{-1/2}\sum_{i=1}^n\gamma_{ihj}(s,t)+o_p(1)$ and \\
$\sqrt{n}[\hat{P}_{n,hj}'(s,t)-P_{0,hj}'(s,t)]=n^{-1/2}\sum_{i=1}^n\gamma_{ihj}'(s,t)+o_p(1)$, $t\in[s,\tau]$. Moreover, the classes of functions $\{\gamma_{ihj}(s,t):t\in[s,\tau]\}$ and $\{\gamma_{ihj}'(s,t):t\in[s,\tau]\}$ are $P$-Donsker.
\item[(ii)] $\hat{B}_{hj}(s,\cdot)\overset{p}{\underset{\xi}\leadsto}\mathbb{G}_{hj}(s,\cdot)$ and $\sqrt{n}[\hat{P}_{n,hj}^*(s,\cdot)-\hat{P}_{n,hj}(s,\cdot)]\overset{p}{\underset{U}\leadsto}\mathbb{G}_{hj}(s,\cdot)$ in $D[s,\tau]$, where $\mathbb{G}_{hj}(s,\cdot)$ is the asymptotic tight limit of the process $\sqrt{n}[\hat{P}_{n,hj}(s,\cdot)-P_{0,hj}(s,\cdot)]$.
\item[(iii)] $\hat{B}_{hj}'(s,\cdot)\overset{p}{\underset{\xi}\leadsto}\mathbb{G}_{hj}'(s,\cdot)$ and $\sqrt{n}[\hat{P}_{n,hj}^{\prime *}(s,\cdot)-\hat{P}_{n,hj}'(s,\cdot)]\overset{p}{\underset{U}\leadsto}\mathbb{G}_{hj}'(s,\cdot)$ in $D[s,\tau]$, where $\mathbb{G}_{hj}'(s,\cdot)$ is the asymptotic tight limit of the process $\sqrt{n}[\hat{P}_{n,hj}'(s,\cdot)-P_{0,hj}'(s,\cdot)]$.
\end{itemize}
\end{theorem}

The proof of Theorem 2 can be found in the Appendix A.3. In Appendix A.5 we consider an alternative weak convergence theorem for situations where condition C6 does not hold. Theorem 2 implies that, for any $h\in\mathcal{T}^c$, $j\in\mathcal{S}$, and $s\in[0,\tau)$, $\sqrt{n}[\hat{P}_{n,hj}(s,\cdot)-P_{0,hj}(s,\cdot)]$ and $\sqrt{n}[\hat{P}_{n,hj}(s,\cdot)-P_{0,hj}(s,\cdot)]$ converge weakly to the tight mean-zero Gaussian processes $\mathbb{G}_{hj}(s,\cdot)$ and $\mathbb{G}_{hj}'(s,\cdot)$, respectively, in $D[s,\tau]$. The asymptotic covariance functions of $\mathbb{G}_{hj}(s,\cdot)$ and $\mathbb{G}_{hj}'(s,\cdot)$ at the time points $t_1$ and $t_2$ are $E[\gamma_{1hj}(s,t_1)\gamma_{1hj}(s,t_2)]$ and $E[\gamma_{1hj}'(s,t_1)\gamma_{1hj}'(s,t_2)]$. These covariance functions can be consistently (in probability) estimated by $n^{-1}\sum_{i=1}^n\hat{\gamma}_{ihj}(s,t_1)\hat{\gamma}_{ihj}(s,t_2)$ and $n^{-1}\sum_{i=1}^n\hat{\gamma}_{ihj}'(s,t_1)\hat{\gamma}_{ihj}'(s,t_2)$, respectively. Theorem 2 also implies that the asymptotic distributions of the estimators can be easily approximated by simulating realizations of the processes $\hat{B}_{hj}(s,\cdot)$ and $\hat{B}_{hj}'(s,\cdot)$, or by bootstrap realizations $\sqrt{n}[\hat{P}_{n,hj}^{*}(s,\cdot)-\hat{P}_{n,hj}(s,\cdot)]$ and $\sqrt{n}[\hat{P}_{n,hj}^{\prime *}(s,\cdot)-\hat{P}_{n,hj}'(s,\cdot)]$. This can be easily performed, conditionally on the observed data, by simulating a large number of sets of standard normal variates $\{\xi_i\}_{i=1}^n$ or multinomial vectors $\mathbf{U}_n$, and then calculating the corresponding realizations of the aforementioned processes.

These results can be used for the calculation of pointwise confidence intervals and simultaneous confidence bands for the transition probabilities. For these procedures it is important to consider a differentiable transformation $g$, such as $g(x)=\log[-\log(x)]$, to ensure that the limits of the confidence interval and the confidence band lie in the interval $(0,1)$. For the calculation of simultaneous confidence bands for $P_{0,hj}(s,\cdot)$, for $h\in\mathcal{T}^c$, $j\in\mathcal{S}$, and $s\in[0,\tau)$, it is useful to consider a weight function $\hat{q}_{hj}(s,t)$ that converges uniformly (in probability) to a bounded non-negative function on an interval $[t_1,t_2]\subset[s,\tau]$. A choice 
is $\hat{q}_{hj}(s,t)=\{1+n^{-1}\sum_{i=1}^n[\hat{\gamma}_{1hj}(s,t)^2]\}^{-1}$, where, as argued above, $n^{-1}\sum_{i=1}^n[\hat{\gamma}_{1hj}(s,\cdot)^2]$ is uniformly consistent for the true asymptotic variance function of $\sqrt{n}[\hat{P}_{n,hj}(s,\cdot)-P_{0,hj}(s,\cdot)]$. By Theorem 2, the functional delta method, and the continuous mapping theorem it follows that the random sequences $
\sup_{t\in[t_1,t_2]}\left|\sqrt{n}\hat{q}_{hj}(s,t)[g(\hat{P}_{n,hj}(s,t))-g(P_{0,hj}(s,t))]\right|$
and
\[
\sup_{t\in[t_1,t_2]}\left|\hat{q}_{hj}(s,t)\dot{g}(P_{0,hj}(s,t))\sqrt{n}[\hat{P}_{n,hj}(s,t))-P_{0,hj}(s,t)]\right|,
\]
have the same limiting distribution. Under Theorem 2, the $1-\alpha$ percentile of this limiting distribution, denoted by $c_{\alpha}$, can be estimated as the sample percentile $\hat{c}_{\alpha}$ of a sufficiently large sample of simulation realizations of $\sup_{t\in[t_1,t_2]}|\hat{q}_{hj}(s,t)\dot{g}(P_{0,hj}(s,t))\hat{B}_{hj}(s,t)|$, or bootstrap realizations $
\sup_{t\in[t_1,t_2]}\left|\hat{q}_{hj}(s,t)\dot{g}(P_{0,hj}(s,t))\sqrt{n}[\hat{P}_{n,hj}^*(s,t)-\hat{P}_{n,hj}(s,t)]\right|$. Based on this $\hat{c}_{\alpha}$, a $1-\alpha$ simultaneous confidence band can be calculated as
\[
g^{-1}\left\{g(\hat{P}_{n,hj}(s,t))\pm\frac{\hat{c}_{\alpha}}{\sqrt{n}\hat{q}_{hj}(s,t)}\right\}, \ \ \ \ t\in[t_1,t_2].
\]
In general, simultaneous confidence bands can be unstable towards the earlier or later times of the observation interval \citep{Nair84}. To avoid this issue in practice we suggest restricting the domain of the confidence band to a set with limits the 10th and 90th or the 5th and 95th percentile of the distribution of unique jump times of the counting processes $N_{im,hj}(t)$. Calculation of simultaneous confidence bands for $P_{0,hj}'(s,\cdot)$, for $h\in\mathcal{T}^c$, $j\in\mathcal{S}$, and $s\in[0,\tau)$ can be performed in a similar manner.

\subsection{Estimation of state occupation probabilities}
\label{ss:SOP}
In many applications, state occupation probabilities are more scientifically relevant compared to transition probabilities. In this subsection we provide estimators and inference procedures for the population-averaged state occupation probabilities. Natural plug-in estimators for the population-averaged state occupation probabilities are 
\[
\hat{P}_{n,j}(t)=\sum_{h\in\mathcal{T}^c}\left\{\frac{\sum_{i=1}^nY_{i\cdot,h}(0+)}{\hat{\pi}_n\sum_{i=1}^nM_i}\right\}\hat{P}_{n,hj}(0,t),  \ \ \ \ j\in\mathcal{S},
\]
 where $\hat{\pi}_n=n^{-1}\sum_{i=1}^nM_i^{-1}\sum_{h\in\mathcal{T}^c}Y_{i\cdot,h}(0+)$, and 
\[
\hat{P}_{n,j}'(t)=\sum_{h\in\mathcal{T}^c}\left\{\frac{\sum_{i=1}^nM_i^{-1}Y_{i\cdot,h}(0+)}{n\hat{\pi}_n}\right\}\hat{P}_{n,hj}'(0,t), \ \ \ \ j\in\mathcal{S}.
\]
In these estimators, $\hat{\pi}_n$ is a consistent estimate of the probability of being under observation at time $t=0$, denoted as $\pi_0$. Here, we also assume that $\pi_0>0$. In the absence of left truncation $\hat{\pi}_n=\pi_0=1$. In the special case with fixed cluster size, $\hat{P}_{n,j}=\hat{P}_{n,j}'$, $j\in\mathcal{S}$. Theorem 1, the continuous mapping theorem, and the strong law of large numbers imply that these state occupation probability estimators are uniformly consistent (outer almost surely) for the corresponding population-averaged state occupation probabilities $P_{0,j}(t)$ and $P_{0,j}'(t)$ over $[0,\tau]$.

The relationship between state occupation and transition probabilities along with Theorem 2 suggest rigorous inference procedures for the former
. It is not hard to see that, in light of Theorem 2, the state occupation probability estimators are asymptotically linear of the form
\[
\sqrt{n}[\hat{P}_{n,j}(t)-P_{0,j}(t)]=\frac{1}{\sqrt{n}}\sum_{i=1}^n\psi_{ij}(t)+o_p(1), \ \ \ \ j\in\mathcal{S}, \ \ t\in[0,\tau],
\]
where 
\begin{eqnarray*}
\psi_{ij}(t)&=&\sum_{h\in\mathcal{T}^c}\bigg(P_{0,h}(0)\gamma_{ihj}(0,t)+P_{0,hj}(0,t)\bigg[\frac{Y_{i\cdot,h}(0+)-EY_{1\cdot,h}(0+)}{\pi_0EM_1} \\
&&-P_{0,h}(0)\bigg\{\frac{M_i-EM_1}{EM_1}+\frac{M_i^{-1}Y_{i\cdot,\cdot}(0+)-\pi_0}{\pi_0}\bigg\}\bigg]\bigg),
\end{eqnarray*}
with $Y_{i\cdot,\cdot}(0+)=\sum_{h\in\mathcal{T}^c}Y_{i\cdot,h}(0+)$, and
\[
\sqrt{n}[\hat{P}_{n,j}'(t)-P_{0,j}'(t)]=\frac{1}{\sqrt{n}}\sum_{i=1}^n\psi_{ij}'(t)+o_p(1)\ \ \ \ j\in\mathcal{S}, \ \ t\in[0,\tau].
\]
with
\begin{eqnarray*}
\psi_{ij}'(t)&=&\sum_{h\in\mathcal{T}^c}(P_{0,h}'(0)\gamma_{ihj}'(0,t)+P_{0,hj}'(0,t)\pi_0^{-1}[M_i^{-1}Y_{i\cdot,h}(0+)-E\{M_1^{-1}Y_{1\cdot,h}(0+)\} \\
&&-P_{0,h}'(0)\{M_1^{-1}Y_{1\cdot,\cdot}(0+)-\pi_0\}]).
\end{eqnarray*}
The classes of functions $\{\psi_{ij}(t):t\in[0,\tau]\}$ and $\{\psi_{ij}'(t):t\in[0,\tau]\}$ are $P$-Donsker for any $j\in\mathcal{S}$. This is due to the fact that these classes consist of linear combinations of functions that belong to $P$-Donsker classes by Theorem 2, fixed functions, and random variables with bounded second moments. Therefore, $\sqrt{n}(\hat{P}_{n,j}-P_{0,j})$ and $\sqrt{n}(\hat{P}_{n,j}'-P_{0,j}')$ converge weakly to tight zero-mean Gaussian processes in $D[0,\tau]$, with covariance functions $E[\psi_{ij}(t_1)\psi_{ij}(t_2)]$ and $E[\psi_{ij}'(t_1)\psi_{ij}'(t_2)]$, for $t_1,t_2\in[0,\tau]$. As with the transition probabilities, the estimated influence functions can be used to consistently (in probability) estimate these covariance functions. The triangle inequality along with results presented in the proof of Theorem 2 in the Appendix A.3 can be easily used to justify the use of the estimated processes $n^{-1/2}\sum_{i=1}^n\hat{\psi}_{ij}(\cdot)\xi_i$ and $n^{-1/2}\sum_{i=1}^n\hat{\psi}_{ij}'(\cdot)\xi_i$ for approximating the asymptotic distributions of $\sqrt{n}(\hat{P}_{n,j}-P_{0,j})$ and $\sqrt{n}(\hat{P}_{n,j}'-P_{0,j}')$, respectively. The validity of the nonparametric cluster bootstrap for the state occupation probabilities follows from Theorem 2 and the bootstrap functional delta method \citep[Theorem 12.1 in][]{Kosorok08}. Therefore, the calculation of simultaneous confidence bands for the state occupation probabilities proceeds as for the case of transition probabilities described in subsection \ref{ss:TP}.

\subsection{Two-sample Kolmogorov--Smirnov-type tests}
\label{ss:test}
In many settings, the scientific interest is on comparing the transition probabilities for a particular transition $h\rightarrow j$ of the process $X(t)$ between two populations, say populations 1 and 2. Depending on what is the most relevant population-averaged quantity for the given context, the null hypothesis is either $H_0:P_{0,1hj}(s,\cdot)=P_{0,2hj}(s,\cdot)$ or $H_0:P_{0,1hj}'(s,\cdot)=P_{0,2hj}'(s,\cdot)$, for some $s\in[0,\tau)$. The corresponding two-sided alternative hypotheses are $H_1:P_{0,1hj}(s,\cdot)\neq P_{0,2hj}(s,\cdot)$ and $H_0:P_{0,1hj}'(s,\cdot)\neq P_{0,2hj}'(s,\cdot)$. Alternatively, one may be interested in comparing the state occupation probabilities for a particular state $j\in\mathcal{S}$ between the two populations. The null hypothesis in this case is either $H_0:P_{0,1j}=P_{0,2j}$ or $H_0:P_{0,1j}'=P_{0,2j}'$. Testing such hypotheses can be based on a sample of clusters of observations of the stochastic process of interest, which satisfies the requirements described in subsection \ref{ss:clust}. An example of such study is a multicenter randomized controlled trial where, for each cluster (e.g. center or clinic), some cluster members receive the intervention of interest and the remaining cluster members receive placebo. For such cases let $M_{1i}$ and $M_{2i}$ to be the number of observations from the $i$th cluster which belong to samples 1 and 2, respectively, with $M_{1i}+M_{2i}=M_i$, $i=1,\ldots,n$. We consider the situation where $M_{1i}\wedge M_{2i}>0$ almost surely. Here, we denote the counting and at-risk processes for the $m$th observation in the $p$th sample in the $i$th cluster as $N_{ipm,hj}(t)$, $h\neq j$, and $Y_{ipm,h}(t)$, $h\in\mathcal{T}^c$.

Based on this setup, define the estimators of the pointwise between-sample difference with respect to the population-averaged transition probabilities as
\[
\hat{\Delta}_{n,hj}(s,t)=\left[\hat{P}_{n,1hj}(s,t)-\hat{P}_{n,2hj}(s,t)\right], \ \ \ \ t\in[s,\tau],
\]
where $\hat{P}_{n,phj}$, $p=1,2$, is the estimator of $P_{0,phj}$ from the $p$th sample and
\[
\hat{\Delta}_{n,hj}'(s,t)=\left[\hat{P}_{n,1hj}'(s,t)-\hat{P}_{n,2hj}'(s,t)\right], \ \ \ \ t\in[s,\tau],
\]
where $\hat{P}_{n,phj}'$, $p=1,2$, is the estimator of $P_{0,phj}'$ from the $p$th sample, for some $s\in[0,\tau)$. Similarly, define the differences between the population-averaged state occupation probabilities as $
\hat{\Delta}_{n,j}(t)=[\hat{P}_{n,1j}(t)-\hat{P}_{n,2j}(t)]$, $t\in[0,\tau]$, where $\hat{P}_{n,pj}$, $p=1,2$, is the estimator of $P_{0,pj}$ from the $p$th sample, and $\hat{\Delta}_{n,j}'(t)=[\hat{P}_{n,1j}'(t)-\hat{P}_{n,2j}'(t)]$, $t\in[0,\tau]$, where $\hat{P}_{n,pj}'$, $p=1,2$, is the estimator of $P_{0,pj}'$ from the $p$th sample. The corresponding nonparametric cluster bootstrap realizations of the above differences are denoted by $\hat{\Delta}_{n,hj}^*(s,t)$, $\hat{\Delta}_{n,hj}^{\prime *}(s,t)$, $\hat{\Delta}_{n,j}^*(t)$, and $\hat{\Delta}_{n,j}^{\prime *}(t)$. It is important to note that these nonparametric cluster bootstrap realizations are generated by randomly sampling $n$ clusters with replacement, as described in subsections \ref{ss:TP} and \ref{ss:SOP}. Based on these differences we define the Kolmogorov--Smirnov-type test statistics $K_{n,hj}(s)=\sup_{t\in[s,\tau]}|\hat{W}_{hj}(t)\hat{\Delta}_{n,hj}(s,t)|$, for some appropriate weight function $\hat{W}_{hj}(t)$ and some $s\in[0,\tau)$, and $K_{n,j}=\sup_{t\in[0,\tau]}|\hat{W}_{j}(t)\hat{\Delta}_{n,j}(t)|$. The corresponding tests for $\hat{\Delta}_{n,hj}'(s,t)$ and $\hat{\Delta}_{n,j}'(t)$, denoted by $K_{n,hj}'(s)$ and $K_{n,j}'$, are defined in the same manner. The weights $\hat{W}_{hj}(t)$, $\hat{W}_{hj}'(t)$, $\hat{W}_{j}(t)$ and $\hat{W}_{j}'(t)$ are assumed to be uniformly consistent (in probability) for the non-negative and uniformly bounded fixed functions $W_{hj}(t)$, $W_{hj}'(t)$, $W_{j}(t)$ and $W_{j}'(t)$. The importance of the weight functions lies on the fact that they can restrict the comparison interval to a set of times where both samples under comparison have non-zero observations at risk for the transition of interest. An example of such a weight function is $\hat{W}_{hj}(t)=I[\prod_{l\in L(h,j)}\bar{Y}_{1,l}(t)\bar{Y}_{2,l}(t)>0]$, where $L(h,j)=\{d\in\mathcal{S}:d$ is a transient state that can be visited during the transition $h\rightarrow j\}$ and $\bar{Y}_{p,h}(t)=n_p^{-1}\sum_{i=1}^{n_p}Y_{pi\cdot,h}(t)$, for the sample $p=1,2$, with $M_{pi}$ denoting the number of observations in the $i$th cluster of the $p$th sample, and $Y_{pi\cdot,h}(t)$ denoting the sum of the at-risk process for state $h$ in the $i$th cluster and the $p$th sample, $p=1,2$. Similarly, this type of weight can be defined for the state occupation probabilities as $\hat{W}_{j}(t)=I[\prod_{l\in\cup_{h\in\mathcal{T}^c}L(h,j)}\bar{Y}_{1,l}(t)\bar{Y}_{2,l}(t)>0]$. The weights $\hat{W}_{hj}'(t)$ and $\hat{W}_{j}'(t)$ are similarly defined. The weight functions can also be used to assign less weight to observation times with a smaller number of observations at risk where the estimated difference tends to be unstable. An example of such weight functions is
\[
\hat{W}_{hj}(t)=\frac{\prod_{l\in L(h,j)}\bar{Y}_{1,l}(t)\bar{Y}_{2,l}(t)}{\sum_{l\in L(h,j)}\left[\bar{Y}_{1,l}(t)+\bar{Y}_{2,l}(t)\right]} \ \ \textrm
{and} \ \ \hat{W}_{j}(t)=\frac{\prod_{l\in\cup_{h\in\mathcal{T}^c}L(h,j)}\bar{Y}_{1,l}(t)\bar{Y}_{2,l}(t)}{\sum_{l\in \cup_{h\in\mathcal{T}^c}L(h,j)}\left[\bar{Y}_{1,l}(t)+\bar{Y}_{2,l}(t)\right]}.
\]
The corresponding weights $\hat{W}_{hj}'(t)$ and $\hat{W}_{j}'(t)$ can be similarly defined by replacing $\bar{Y}_{p,h}(t)$ with $n_p^{-1}\sum_{i=1}^{n_p}M_{pi}^{-1}Y_{pi\cdot,h}(t)$, for the sample $p=1,2$. In practice we suggest the use of this latter type of weight functions. The asymptotic distribution of the Kolmogorov--Smirnov-type tests under the null hypothesis is not normal and has a complicated form as it will be shown later. However, we will show that conducting hypothesis testing with these tests can be based either on nonparametric cluster bootstrap or estimated processes similar to those defined for the construction of simultaneous confidence bands in subsection \ref{ss:TP}. For the latter case, consider the influence functions for the sample-specific estimators $\hat{P}_{n,phj}(s,t)$ and $\hat{P}_{n,pj}(t)$, $p=1,2$, which are denoted by $\gamma_{p,ihj}(s,t)$ and $\psi_{p,ij}(t)$, $p=1,2$, and are defined as in subsections \ref{ss:TP} and \ref{ss:SOP}. Now, define the estimated processes $\hat{C}_{n,hj}(s,t)=\hat{W}_{hj}(t)n^{-1/2}\sum_{i=1}^{n}[\hat{\gamma}_{1,ihj}(s,t)-\hat{\gamma}_{2,ihj}(s,t)]\xi_{i}$, $t\in[s,\tau]$, for some $s\in[0,\tau)$, where $\xi_{i}$, are independent standard normal variables and the influence functions are estimated as described in subsection \ref{ss:TP}, and $\hat{C}_{n,j}(t)=\hat{W}_{j}(t)n^{-1/2}\sum_{i=1}^{n}[\hat{\psi}_{1,ij}(t)-\hat{\psi}_{2,ij}(t)]\xi_{i}$, $t\in[0,\tau]$. Similarly, one can define the estimated processes $\hat{C}_{n,hj}(s,t)$ and $\hat{C}_{n,j}'(t)$ which correspond to the tests for $\hat{\Delta}_{n,hj}'(s,t)$ and $\hat{\Delta}_{n,j}'(t)$. Theorem 3 provides the basis for conducting two-sample testing
.

\begin{theorem}
Suppose that conditions C1, C2, C3', C4', C5 and C6' in Appendix A.1 hold. Then, under the null hypothesis and for any $h\in\mathcal{T}^c$, $j\in\mathcal{S}$, and $s\in[0,\tau)$,
\begin{itemize}
\item[(i)] $\sqrt{n}\hat{W}_{hj}(\cdot)\hat{\Delta}_{n,hj}(s,\cdot)\leadsto\mathbb{Z}_{hj}(s,\cdot)$ in $D[s,\tau]$, where $\mathbb{Z}_{hj}(s,\cdot)$ is a tight zero-mean Gaussian process with covariance function
\[
W_{hj}(t_1)W_{hj}(t_2)E\{[\gamma_{1,1hj}(s,t_1)-\gamma_{2,1hj}(s,t_1)][\gamma_{1,1hj}(s,t_2)-\gamma_{2,1hj}(s,t_2)]\},
\]
for $t_1,t_2\in[s,\tau]$. Moreover, $\hat{C}_{n,hj}(s,\cdot)\overset{p}{\underset{\xi}\leadsto}\mathbb{Z}_{hj}(s,\cdot)$ in $D[s,\tau]$, and
\[
\sqrt{n}\hat{W}_{hj}(\cdot)[\hat{\Delta}_{n,hj}^*(s,\cdot)-\hat{\Delta}_{n,hj}(s,\cdot)]\overset{p}{\underset{U}\leadsto}\mathbb{Z}_{hj}(s,\cdot) \ \ \textrm{in} \ \ D[s,\tau].
\]
\item[(ii)] $\sqrt{n}\hat{W}_{hj}\hat{\Delta}_{n,j}\leadsto\mathbb{Z}_{j}$ in $D[0,\tau]$, where $\mathbb{Z}_{j}$ is a tight zero-mean Gaussian process with covariance function
\[
W_{j}(t_1)W_{j}(t_2)E\{[\psi_{1,1j}(s,t_1)-\psi_{2,1j}(s,t_1)][\psi_{1,1j}(s,t_2)-\psi_{2,1j}(s,t_2)]\},
\]
for $t_1,t_2\in[s,\tau]$. Moreover, $\hat{C}_{n,j}\overset{p}{\underset{\xi}\leadsto}\mathbb{Z}_{j}$ in $D[0,\tau]$, and
\[
\sqrt{n}\hat{W}_{j}(\hat{\Delta}_{n,j}^*-\hat{\Delta}_{n,j})\overset{p}{\underset{U}\leadsto}\mathbb{Z}_{j} \ \ \textrm{in} \ \ D[0,\tau].
\]
\end{itemize}
\end{theorem}

The proof of Theorem 3 can be found in Appendix A.4. A relaxation of condition C6' is also presented in Appendix A.5. Using the same arguments given in this proof, it can be shown that a similar version of Theorem 3 holds for the differences $\hat{\Delta}_{h,hj}'(s,\cdot)$ and $\hat{\Delta}_{h,j}'$. Based on Theorem 3 and the continuous mapping theorem it follows that, under the null hypothesis, $\sqrt{n}K_{n,hj}(s)\overset{d}\rightarrow\sup_{t\in[s,\tau]}|\mathbb{Z}_{hj}(s,t)|$, for any $s\in[0,\tau)$, and $\sqrt{n}K_{n,j}\overset{d}\rightarrow\sup_{t\in[0,\tau]}|\mathbb{Z}_{j}(t)|$. These asymptotic null distributions are complicated to use in practice for the calculation of $p$-values. However, by Theorem 3 and the continuous mapping theorem, one can simulate realizations from these null distributions by simulating a sufficiently large number of sets $\{\xi_{i}\}_{i=1}^{n}$ of independent standard normal variables and then calculating samples from these null distributions as $\sup_{t\in[s,\tau]}|\hat{C}_{n,hj}(s,t)|$ and $\sup_{t\in[0,\tau]}|\hat{C}_{n,j}(t)|$. Alternatively, realizations from these asymptotic null distributions can be generated by obtaining a sufficiently large number of nonparametric cluster bootstrap realizations $\hat{\Delta}_{n,hj}^*(s,t)$, $t\in[s,\tau]$, and $\hat{\Delta}_{n,j}^*(t)$, $t\in[0,\tau]$. Then, simulation realizations from the asymptotic null distributions can be calculated as $\sqrt{n}\sup_{t\in[s,\tau]}|\hat{W}_{hj}(t)[\hat{\Delta}_{n,hj}^*(s,t)-\hat{\Delta}_{n,hj}(s,t)]|$ and $\sqrt{n}\sup_{t\in[0,\tau]}|\hat{W}_{j}(t)[\hat{\Delta}_{n,j}^*(t)-\hat{\Delta}_{n,j}(t)]|$. The $p$-value can then be estimated as the proportion of simulation realizations from the corresponding asymptotic null distribution which are greater than or equal to the actual value of the test statistic based on the observed data. The proposed Kolmogorov--Smirnov-type tests are consistent. 
This follows from Theorem 3, the uniform consistency of the proposed transition probability and state occupation probability estimators, the continuity of these tests in the differences $\hat{\Delta}_{n,hj}(s,t)$, $\hat{\Delta}_{n,j}(t)$, $\hat{\Delta}_{n,hj}'(s,t)$, and $\hat{\Delta}_{n,j}'(t)$, and Lemma 14.15 in \citet{Van00}.

\subsection{Non-Markov processes}
\label{ss:nonmark}
When the stochastic process $X(t)$ is non-Markov, the transition probabilities and transition intensities depend on the prior history $\mathcal{F}_{t^-}$.  In this case, the population-averaged transition intensities defined in subsection \ref{ss:clust} are the \textit{partly condition transition intensities} \citep{Pepe93,Datta01,Glidden02}, which are not conditional on the prior history $\mathcal{F}_{t^-}$. Such marginal intensities have been argued to be meaningful quantities even for non-Markov processes, because they describe the marginal (i.e. unconditional on the prior history) behavior of the process 
\citep{Datta01,Glidden02}. With independent observations from a non-Markov process, \citet{Datta01} showed that the Nelson--Aalen estimator of the cumulative transition intensities and the Aalen--Johansen estimator of the state occupation probabilities are consistent for the corresponding marginal quantities. Using the same arguments to those presented by \citet{Datta01} it can be shown that, with clustered observations from a non-Markov process, the proposed estimators of the (marginal) population-averaged cumulative transition intensities and state occupation probabilities are consistent. Similarly, as in the case with independent observations \citep{Titman15}, the proposed estimators $\hat{\mathbf{P}}_{n}(0,t)$ and $\hat{\mathbf{P}}_{n}'(0,t)$ are consistent for the population-averaged $\mathbf{P}_{0}(0,t)$ and $\mathbf{P}_{0}'(0,t)$, even for non-Markov processes. However, for $s>0$, the proposed estimators $\hat{\mathbf{P}}_{n}(s,t)$ and $\hat{\mathbf{P}}_{n}'(s,t)$ are not consistent, in general, for non-Markov processes, as in the case with independent observations \citep{Titman15}. In such cases, following \citet{Putter18}, we propose the landmark version of the working-independence and weighted by cluster size working-independence Aalen--Johansen estimators. This estimator can be obtained using the modified counting and at-risk processes $\tilde{N}_{im,hj}(t)=N_{im,hj}(t)I(X_{im}(s)=h)$ and $\tilde{Y}_{im,h}(s)=Y_{im,h}(t)I(X_{im}(s)=h)$, instead of the original $N_{im,hj}(t)$ and $Y_{im,h}(t)$, in $\hat{\mathbf{P}}_{n}(s,t)$ and $\hat{\mathbf{P}}_{n}'(s,t)$. The landmark versions of the proposed transition probability estimators can be shown to be consistent using the same arguments to those used in \citet{Putter18}.

Inference for the proposed estimators of the marginal population-averaged quantities can be performed as indicated in Theorems 2 and 3, with the exception that the influence functions for the landmark versions of $\hat{\mathbf{P}}_{n}(s,t)$ and $\hat{\mathbf{P}}_{n}'(s,t)$ involve the modified processes $\tilde{N}_{im,hj}(t)$ and $\tilde{Y}_{im,j}(t)$. Note that, the influence functions for the estimators $\hat{P}_{n,hj}(0,t)$, $\hat{P}_{n,hj}'(0,t)$, $\hat{P}_{n,j}(t)$, and  $\hat{P}_{n,j}'(t)$ involve the quantities $P_{0,qj}(u,t)$ and $P_{0,qj}'(u,t)$, for $u>0$. With non-Markov processes, these quantities are defined as the $(q,j)$ element of the matrices $\prodi_{(u,t]}[\mathbf{I}_k+d\mathbf{A}_0(s)]$ and $\prodi_{(u,t]}[\mathbf{I}_k+d\mathbf{A}_0'(s)]$, respectively. The latter matrices are not necessarily equal to the true (conditional on the prior history) transition probability matrices under a non-Markov process. Nevertheless, the true influence functions of the estimators depend on these matrices regardless of the Markov assumption. This is because, given the consistency of the estimators, the derivation of the influence functions in the proof of Theorem 2 (Appendix A.3) does not utilize the Markov assumption. The same phenomenon is observed for the independent observations setting \citep{Glidden02}. Since these matrices are continuous in $\mathbf{A}_0(s)$ and $\mathbf{A}_0'(s)$ \citep{Andersen12}, they can be consistently estimated by $\prodi_{(u,t]}[\mathbf{I}_k+d\hat{\mathbf{A}}_n(s)]$ and $\prodi_{(u,t]}[\mathbf{I}_k+d\hat{\mathbf{A}}_n'(s)]$, in order to estimate the corresponding influence functions.

\section{Simulation studies}
\label{s:sims}

To evaluate the finite sample properties of the proposed methods we conducted a series of simulation experiments under a non-Markov illness-death model with state space $\mathcal{S}=\{1,2,3\}$ and absorbing state space $\mathcal{T}=\{3\}$, in a study with informative cluster size. The goal of these simulation studies was to conduct inference about the population-averaged state occupation probability $P_{0,2}'(t)$. Note that, for the illness-death model where state 1 (healthy) is the unique inital state, $P_{0,2}'(t)=P_{0,12}'(0,t)$. We considered scenarios with $n=20,40,80$ clusters. These sample sizes are considered small or relatively small. The cluster sizes $M_i$, $i=1,\ldots,n$, were simulated  from either of the discrete uniform distributions $\mathcal{U}(5,15)$ and $\mathcal{U}(10,30)$, producing scenarios with 5 to 15 and 10 to 30 observations per cluster, respectively. To simulate non-Markov illness-death processes which are correlated within clusters, we simulated cluster-specific frailties $v_i$, $i=1,\ldots,n$, from the Gamma distribution with shape and scale parameters equal to 1. Conditionally on the frailty values $v_i$ and the cluster sizes $m_i$, we simulated the non-Markov illness-death processes based on the cumulative transition intensities $A_{0,12}'(t;v_i)=\{0.25+0.25\times I[m_i\leq E(M_1)]\}v_it$, $A_{0,23}'(t;v_i)=0.5v_it$, and $A_{0,13}'(t;v_i)=0.25v_it$, $i=1,\ldots,n$. Note that the dependence of $A_{0,12}'(t;v_i)$ on cluster size produced data with informative cluster size. Additionally, independent right censoring times were simulated from the uniform distribution $U(0,3)$. This settings led to 57.5\% right-censored observations, 24.4\% observations at the illness state (2; 45.9\% of those arrived later at the death state), and 18.1\% at the death state (3) without a prior visit to the illness state. Under this setup, we also simulated a two-arm multicenter randomized controlled trial under $H_0:P_{0,12}'=P_{0,22}'$, where $P_{0,p2}'$ denotes the state occupation probability for the $p$th arm, $p=1,2$, with a 1:1 arm allocation ratio within clusters. Data under $H_0:P_{0,12}'\neq P_{0,22}'$ were simulated similarly with the exception that we assumed the intensity $A_{0,p12}'(t;v_i)=\{0.25+0.5\times I(p=2)+0.25\times I[m_i\leq E(M_1)]\}v_it$, $p=1,2$, depending on the treatment arm $p$. Data in all scenarios were analyzed using the proposed methods. Simultaneous confidence bands and $p$-values from the Kolmogorov--Smirnov-type tests were based on 1000 simulated sets $\{\xi_i\}_{i=1}^n$ of standard normal variates or 1000 nonparametric cluster bootstrap realizations. Moreover, as described in subsection \ref{ss:TP}, the range of the confidence bands was restricted for each data set to the 10th and 90th percentile of the distribution of transition times from state 1 to state 2. We also present simulation results for the one-sample case under the working-independence Aalen-Johansen estimator using the usual Greenwood standard error estimates and a wild bootstrap approach for confidence bands that ignores the within-cluster dependence.

Pointwise simulation results under the one-sample setup are presented in Tables \ref{t:point1} and \ref{t:point2}. Ignoring the within-cluster dependence was associated with underestimated standard errors and poor coverage probabilities of the 95\% confidence intervals. Also, the working-independence Aalen-Johansen estimator of $P_{0,2}'(t)$ exhibited some bias as a result of the informative cluster-size. The proposed estimator of $P_{0,2}'(t)$ was virtually unbiased, the standard error estimates based on the influence functions and the nonparametric cluster bootstrap were both close to the Monte Carlo standard deviation of the estimates, and the corresponding 95\% pointwise confidence intervals were close to the nominal level, except for the case with a very small number of clusters ($n$=20) and only 5-15 patients per cluster.

\begin{table}
\caption{Simulation results for the analysis of $P_{0,2}'(\tau_{0.4})$, where $\tau_{0.4}$ is the 40th percentile of the follow-up time, based on the standard approach which ignores the within-cluster dependence (na\"ive) and the proposed method with the estimated process $\hat{C}_{n,j}$ (IF) or the nonparametric cluster bootstrap (CB). ($n$: number of clusters; $F_M$: discrete uniform distribution of the cluster size; $^*$: $\times 10^2$; MCSD: Monte Carlo standard deviation of the estimates; ASE: average estimated standard error; CP: coverage probability).}
\label{t:point1}
\begin{center}
\begin{tabular}{lllcccc}
\hline
$n$ & $F_M$ & Method & Bias$^*$ & MCSD$^*$ & ASE$^*$ & CP \\ \hline
20&$\mathcal{U}[5,15]$&Na\"ive&-1.022&3.226&2.623&0.859 \\
&&Proposed (IF)&-0.063&3.517&3.311&0.926 \\
&&Proposed (CB)&-0.063&3.517&3.316&0.923 \\[1.5ex]
&$\mathcal{U}[10,30]$&Na\"ive&-0.928&2.558&1.855&0.816 \\
&&Proposed (IF)&0.077&2.787&2.702&0.940 \\
&&Proposed (CB)&0.077&2.787&2.698&0.939 \\[1.5ex]
40&$\mathcal{U}[5,15]$&Na\"ive&-0.939&2.199&1.863&0.866 \\
&&Proposed (IF)&0.080&2.403&2.411&0.948 \\
&&Proposed (CB)&0.080&2.403&2.407&0.947 \\[1.5ex]
&$\mathcal{U}[10,30]$&Na\"ive&-1.003&1.808&1.310&0.779 \\
&&Proposed (IF)&-0.012&1.940&1.941&0.946 \\
&&Proposed (CB)&-0.012&1.940&1.940&0.945 \\[1.5ex]
80&$\mathcal{U}[5,15]$&Na\"ive&-1.083&1.551&1.312&0.820 \\
&&Proposed (IF)&-0.055&1.699&1.715&0.940 \\
&&Proposed (CB)&-0.055&1.699&1.711&0.940 \\[1.5ex]
&$\mathcal{U}[10,30]$&Na\"ive&-0.962&1.286&0.928&0.732 \\
&&Proposed (IF)&0.025&1.399&1.382&0.944 \\
&&Proposed (CB)&0.025&1.399&1.382&0.946 \\
\hline
\end{tabular}
\end{center}
\end{table}

\begin{table}
\caption{Simulation results for the analysis of $P_{0,2}'(\tau_{0.6})$, where $\tau_{0.6}$ is the 60th percentile of the follow-up time, based on the standard approach which ignores the within-cluster dependence (na\"ive) and the proposed method with the estimated process $\hat{C}_{n,j}$ (IF) or the nonparametric cluster bootstrap (CB). ($n$: number of clusters; $F_M$: discrete uniform distribution of the cluster size; $^*$: $\times 10^2$; MCSD: Monte Carlo standard deviation of the estimates; ASE: average estimated standard error; CP: coverage probability).}
\label{t:point2}
\begin{center}
\begin{tabular}{lllcccc}
\hline
$n$ & $F_M$ & Method & Bias$^*$ & MCSD$^*$ & ASE$^*$ & CP \\ \hline
20&$\mathcal{U}[5,15]$&Na\"ive&-0.939&3.656&3.033&0.888 \\
&&Proposed (IF)&0.077&3.963&3.651&0.924 \\
&&Proposed (CB)&0.077&3.963&3.663&0.920 \\[1.5ex]
&$\mathcal{U}[10,30]$&Na\"ive&-0.940&2.740&2.140&0.854 \\
&&Proposed (IF)&0.078&2.978&2.899&0.935 \\
&&Proposed (CB)&0.078&2.978&2.899&0.933 \\[1.5ex]
40&$\mathcal{U}[5,15]$&Na\"ive&-1.060&2.364&2.140&0.896 \\
&&Proposed (IF)&0.027&2.592&2.635&0.953 \\
&&Proposed (CB)&0.027&2.592&2.636&0.953 \\[1.5ex]
&$\mathcal{U}[10,30]$&Na\"ive&-1.020&1.943&1.509&0.818 \\
&&Proposed (IF)&-0.011&2.100&2.075&0.937 \\
&&Proposed (CB)&-0.011&2.100&2.075&0.936 \\[1.5ex]
80&$\mathcal{U}[5,15]$&Na\"ive&-1.152&1.738&1.510&0.845 \\
&&Proposed (IF)&-0.084&1.894&1.885&0.949 \\
&&Proposed (CB)&-0.084&1.894&1.885&0.948 \\[1.5ex]
&$\mathcal{U}[10,30]$&Na\"ive&-0.972&1.433&1.070&0.775 \\
&&Proposed (IF)&0.045&1.543&1.487&0.942 \\
&&Proposed (CB)&0.045&1.543&1.488&0.945 \\
\hline
\end{tabular}
\end{center}
\end{table}

Simulation results regarding the coverage probabilities of the 95\% simultaneous confidence bands are presented in Table \ref{t:bands}. The wild bootstrap approach for confidence band calculation, which ignores the within-cluster dependence, exhibited poor coverage rates. On the contrary, the coverage probabilities of the proposed approaches were close to the nominal level, except for the case with 20 clusters (smallest cluster size) and 5-15 observations per cluster, where the coverage rate was somewhat lower. Finally, simulation results about the empirical rejection rates of the proposed Kolmogorov--Smirnov-type tests are presented in Table \ref{t:tests}. Under $H_0$, the type I error rate of the tests was close to the nominal level $\alpha=0.05$ in all cases. Under $H_1$, the empirical power was increasing with sample size and this provides numerical evidence for the consistency of the proposed tests.

\begin{table}
\caption{Simulation results regarding the coverage probabilities of the 95\% simultaneous confidence bands based on the standard method that ignores the within-cluster dependence (na\"ive) and the proposed method with the estimated process $\hat{C}_{n,j}$ (IF) or the nonparametric cluster bootstrap (CB). ($n$: number of clusters; $F_M$: discrete uniform distribution of the cluster size).}
\label{t:bands}
\begin{center}
\begin{tabular}{llccc}
\hline
$n$ & $F_M$ &  Na\"ive & \multicolumn{2}{c}{Proposed} \\
&&& IF & CB \\ \hline
20&$\mathcal{U}[5,15]$&0.826&0.917&0.911 \\
&$\mathcal{U}[10,30]$&0.771&0.946&0.938 \\
40&$\mathcal{U}[5,15]$&0.849&0.945&0.940 \\
&$\mathcal{U}[10,30]$&0.750&0.945&0.946 \\
80&$\mathcal{U}[5,15]$&0.788&0.940&0.942 \\
&$\mathcal{U}[10,30]$&0.689&0.945&0.940 \\
\hline
\end{tabular}
\end{center}
\end{table}

\begin{table}
\caption{Simulation results regarding the empirical type I error ($H_0$) and the empirical power ($H_1$) of the proposed Kolmogorov--Smirnov-type tests at the $\alpha=0.05$ level. Significance levels were calculated based on either the estimated processes $\hat{C}_{n,j}$ (IF) or the nonparametric cluster bootstrap (CB). ($n$: number of clusters; $F_M$: distribution of the cluster size).}
\label{t:tests}
\begin{center}
\begin{tabular}{llcccc}
\hline
&& \multicolumn{2}{c}{$H_0$} &  \multicolumn{2}{c}{$H_1$} \\
$n$ & $F_M$ & IF & CB & IF & CB \\ \hline
20&$\mathcal{U}[5,15]$&0.049&0.050&0.331&0.337 \\
&$\mathcal{U}[10,30]$&0.044&0.040&0.598&0.601 \\
40&$\mathcal{U}[5,15]$&0.037&0.039&0.612&0.603 \\
&$\mathcal{U}[10,30]$&0.044&0.046&0.874&0.873 \\
80&$\mathcal{U}[5,15]$&0.049&0.047&0.870&0.864 \\
&$\mathcal{U}[10,30]$&0.059&0.055&0.991&0.990 \\
\hline
\end{tabular}
\end{center}
\end{table}

\section{Analysis of the multicenter EORTC trial 10854}
\label{s:app}
The proposed methods were applied to analyze data from the EORTC trial 10854 \citep{EORTC}. This was a multicenter randomized controlled trial which was conducted to compare the effectiveness of the combination of surgery plus polychemotherapy versus surgery alone as treatment options for early breast cancer. In total, 2793 early breast cancer patients from 15 hospitals (clusters) were recruited in this trial. Of them, 1398 (50.1\%) were randomly assigned to the group receiving the combination therapy approach. In this multicenter trial, cluster sizes ranged from 6 to 902 patients. After surgery, 385 (13.8\%) patients experienced locoregional relapse and 810 (29.0\%) died throughout the follow-up period. This patient event history can be described by an illness-death model with the states ``cancer-free'' (state 1), ``cancer'' (state 2), and ``death'' (state 3). In this analysis we focus on the between-arm comparison of the population-averaged state occupation probabilities of cancer $P_{0,12}(t)$ (for the population undergoing surgery only) and $P_{0,22}(t)$ (for the population receiving the combination of surgery plus polychemotherapy). These population-averaged probabilities correspond to the population of all hospital patients. In this application we consider these estimands more relevant compared to the population-averaged state occupation probabilities $P_{0,12}'(t)$ and $P_{0,22}'(t)$ of typical hospital patients. The overall state occupation probability estimates for the three states, along with the associated 95\% simultaneous confidence bands are presented in Figure \ref{f:sop}. These confidence bands were calculated based on 1000 nonparametric cluster bootstrap realizations. Figure \ref{f:sop} provides significant information about the natural history of early breast cancer patients undergoing surgery. The arm-specific state occupation probabilities of cancer are presented in Figure \ref{f:test}. To compare these population-averaged probabilities between arms we used the proposed Kolmogorov--Smirnov-type test based on 1000 nonparametric cluster bootstrap realizations. This test was not statistically significant ($p$-value=0.097) at the level $\alpha=0.05$ and, therefore, we cannot reject the null hypothesis that the population-averaged probabilities of cancer do not differ between arms. 

\begin{figure}
\begin{center}
\centerline{\includegraphics[width=17cm]{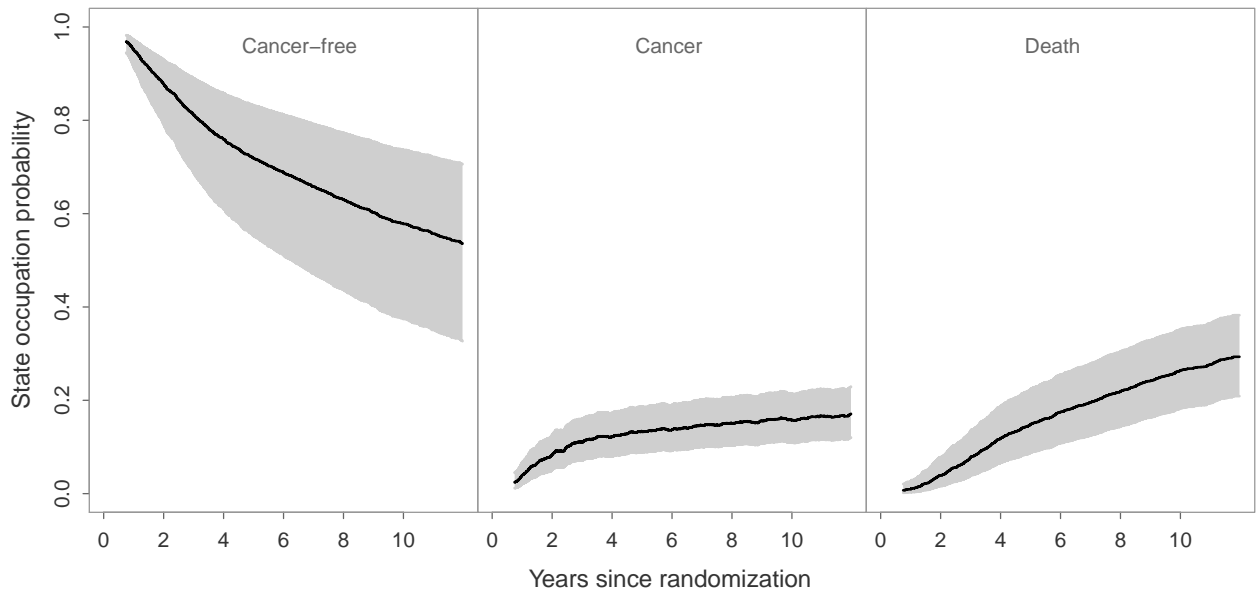}}
\end{center}
\caption{Overall population-averaged state occupation probabilities of the three states (black lines) in the multicenter EORTC trial 10854, along with the 95\% simultaneous confidence bands (grey areas).
\label{f:sop}}
\end{figure}

\begin{figure}
\begin{center}
\centerline{\includegraphics[width=12cm]{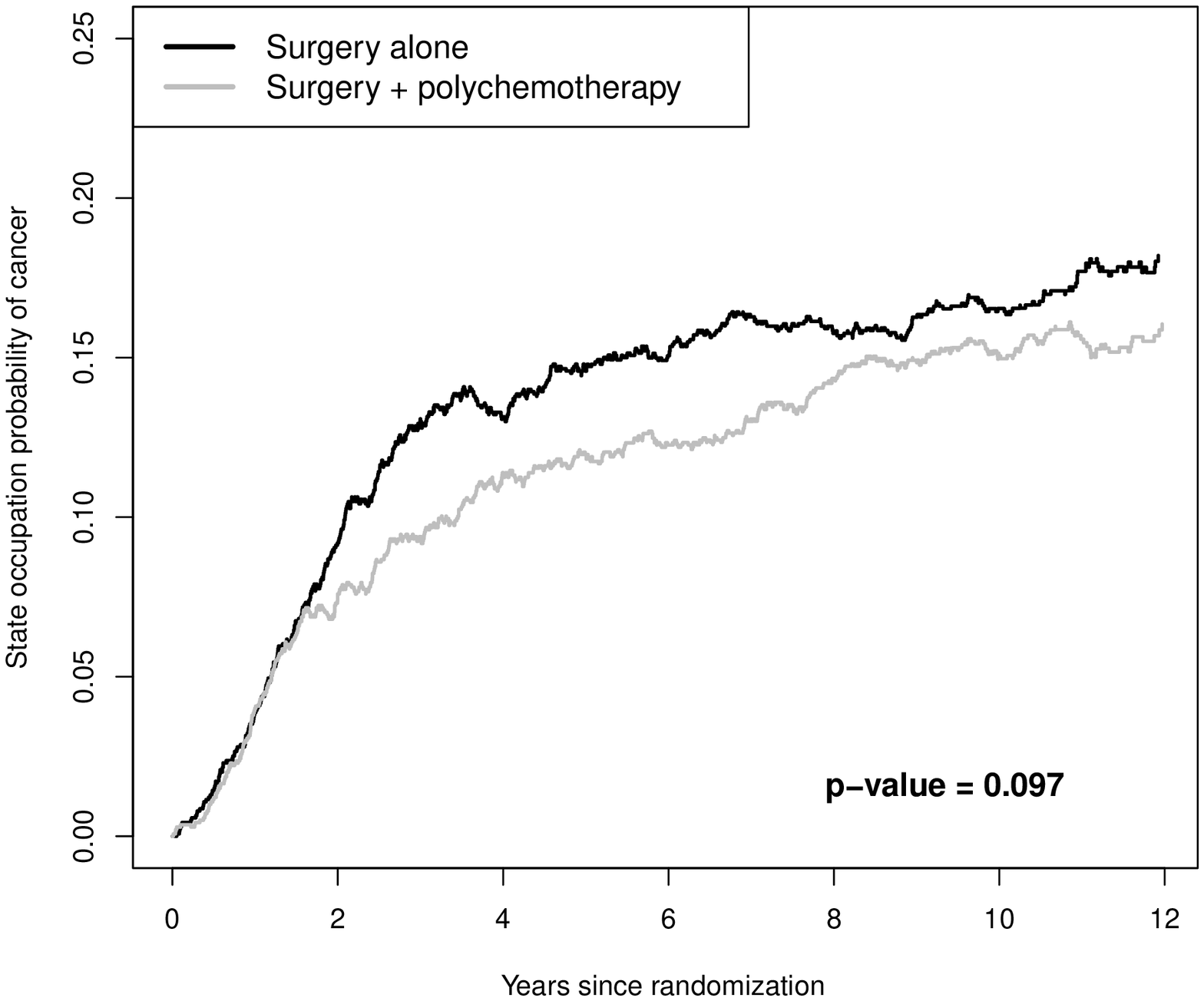}}
\end{center}
\caption{Population-averaged state occupation probabilities of cancer (locoregional relapse, distant metastasis or secondary cancer) for the two arms in the multicenter EORTC trial 10854, along with the $p$-value from the Kolmogorov--Smirnov-type test.
\label{f:test}}
\end{figure}

\section{Discussion}
\label{s:discuss}

In this work we addressed the issue of nonparametric population-averaged inference for multi-state models based on right-censored and/or left-truncated clustered observations. Our estimators for the transition and state occupation probabilities were shown to be uniformly consistent and to converge weakly to tight Gaussian processes with explicit formulas for the corresponding covariance functions. Additionally, we proposed rigorous methodology for the calculation of simultaneous confidence bands and a class of Kolmogorov--Smirnov-type tests. Inference can be performed using either the explicit formulas for the influence functions of the estimators or the nonparametric cluster bootstrap. The latter is particularly useful in practice since it can be used for inference using off-the-shelf software. In this work we did not impose restrictive distributional assumptions or assumptions regarding the within-cluster dependence. Moreover, we allowed for informative cluster size and nonhomogeneous multi-state processes which are non-Markov. Simulation results indicated that the performance of the proposed methods is satisfactory even for non-Markov processes and under an informative cluster size. On the contrary, ignoring the within cluster-dependence lead to invalid inference.

The issue of nonparametric inference for general multi-state models based on clustered observations has not received much attention regardless of its practical importance. So far, to the best of our knowledge, only \citet{Dipankar17} and \citet{OKeeffe18} have provided solutions to this problem. However, the former approach is for current status data and not the usual right-censored or left-truncated multi-state data, and the latter provides cluster-specific inference which may not be of scientific interest in many applications. Moreover, the asymptotic properties of these methods have not been established and, also, there is no methodology for simultaneous confidence bands and nonparametric tests. Our work has addressed this significant gap in the literature of multi-state models. 

We can see two extensions of the proposed framework that will be useful in medical research. First, many studies that use electronic health record data involve incomplete state ascertainment \citep[see e.g.][]{Bakoyannis19}. For such situations it would be useful to propose appropriate extensions of the proposed method. Second, multi-state processes are frequently observed at a particular time point only \citep[current status data, see e.g.][]{Dipankar17} or at a panel of discrete observation times. Adapting the proposed methodology to such observation schemes would be particularly useful. However, we expect that the rate of convergence of such nonparametric estimators will be slower than $\sqrt{n}$, and this would make inference more difficult.

\section*{Acknowledgements}
This project was supported by the National Institute Of Allergy And Infectious Diseases grant number R21AI145662 and the Indiana Clinical and Translational Sciences Institute funded, in part by Grant Number UL1TR002529 from the National Institutes of Health, National Center for Advancing Translational Sciences, Clinical and Translational Sciences Award. We would like to thank the European Organisation for Research and Treatment of Cancer (EORTC) for sharing with us the data from the EORTC trial 10854. The content of this manuscript is solely the responsibility of the authors and does not necessarily represent the official views of the National Institutes of Health and the EORTC.\vspace*{-8pt}

\bibliographystyle{chicago}
\bibliography{references}

\begin{appendices}

\section*{Appendix A: Asymptotic Theory Proofs}
\renewcommand{\thesubsection}{A.\arabic{subsection}}

The proofs of the theorems provided in Section 2 of the manuscript rely on empirical process theory \citep{Van96,Kosorok08}. In this Appendix we use the standard empirical processes notation
\[
\mathbb{P}_nf=\frac{1}{n}\sum_{i=1}^nf({\bf{D}}_i), \textrm{ and } Pf=\int_{\mathcal{D}} fdP=Ef,
\]
where, for any measurable function $f:\mathcal{D}\mapsto\mathbb{R}$, ${\bf{D}}_i$ denotes the observed variables for the $i$th cluster, $\mathcal{D}$ denotes the sample space, and $P$ the true (induced) probability measure defined on the Borel $\sigma$-algebra on $\mathcal{D}$. We also use the supremum norm notation $\|f(t)\|_{\infty}\equiv\sup_{t\in[0,\tau]}|f(t)|$. Let $V$ be a generic constant that may differ from place to place. In this Appendix we only prove the asymptotic properties of $\hat{\mathbf{P}}_n(s,t)$ since the properties of $\hat{\mathbf{P}}_n'(s,t)$, $\hat{P}_{n,j}(t)$, and $\hat{P}_{n,j}'(t)$, $j\in\mathcal{S}$, can be established using the same arguments. Without loss of generality and for simplicity of presentation we set the starting point $s=0$. Before outlining the proofs of Theorems 1-3 we provide and prove two useful lemmas.

\begin{lemma}
Let $N(t)$ be an arbitrary counting process on $[0,\tau]$ with $P[N(\tau)]^2<\infty$ and $h(t)$ be a fixed and non-negative function with $h(t)\leq V$ almost everywhere with respect to the Lebesgue-Stieltjes measure generated by (the sample paths of) $N(t)$. Then, the class of functions
\[
\mathcal{F}_1(s)=\left\{\int_{s}^t h(u) dN(u) : t\in[s,\tau]\right\},
\]
is $P$-Donsker for any $s\in[0,\tau)$.
\end{lemma}
\begin{proof}
Let $\|h\|_{Q,2}=(\int |h|^2dQ)^{1/2}$ for any probability measure $Q$. Now, for any probability measure $Q$ and any $t_1,t_2\in[0,\tau]$ it follows that
\begin{eqnarray*}
\left\|\int_{s}^{t_1}h(u)dN(u) - \int_{s}^{t_2}h(u)dN(u) \right\|_{Q,2}&\leq& \left\|\int_{t_1}^{t_2}h(u)dN(u)\right\|_{Q,2}  \\
&\leq& V\|N(t_2)-N(t_1)\|_{Q,2}.
\end{eqnarray*}
By lemma 22.4 in \citet{Kosorok08}, it follows that the class $\Phi_1=\{N(t):t\in[0,\tau]\}$ has a bounded uniform entropy integral (BUEI) with envelope $2N(\tau)$, and is also pointwise measurable (PM). This implies that, for any $t\in[0,\tau]$ there exist a $t_i\in[0,\tau]$, $i=1,\ldots,N(\epsilon2\|N(\tau)\|_{Q,2},\Phi_1,L_2(Q))$, such that $\|N(t)-N(t_i)\|_{Q,2}<\epsilon2\|N(\tau_2)\|_{Q,2}$, for any $\epsilon>0$ and any finitely discrete probability measure $Q$. Therefore, for any member of $\mathcal{F}_1(s)$, there exist a $\int_{s}^{t_i}h(u)dN(u)$, for $i=1,\ldots,N(\epsilon 2\|N(\tau)\|_{Q,2},\Phi_1,L_2(Q))$, such that
\[
\left\|\int_{s}^{t}h(u)dN(u) - \int_{s}^{t_i}h(u)dN(u) \right\|_{Q,2}\leq \epsilon 2V\|N(\tau)\|_{Q,2},
\]
for any $\epsilon>0$ and any finitely discrete probability measure $Q$. Consequently, by the minimality of the covering number it follows that for any $\epsilon>0$ and any finitely discrete probability measure $Q$, we have that
\[
N(\epsilon 2V\|N(\tau)\|_{Q,2},\mathcal{F}_1(s),L_2(Q))\leq N(\epsilon2\|N(\tau)\|_{Q,2},\Phi_1,L_2(Q)),
\]
which yields a BUEI for $\mathcal{F}_1(s)$ with envelope $2VN(\tau)$. Using similar arguments to those used in the example of page 142 of \citet{Kosorok08}, it can be shown that the class $\mathcal{F}_1(s)$ is also PM. Therefore, by Theorem 2.5.2 in \citet{Van96}, the class $\mathcal{F}_1(s)$ is $P$-Donsker. Since $s$ was arbitrary, the last statement is true for any $s\in[0,\tau)$.
\end{proof}

\begin{lemma}
Let $Y(t)$ be an arbitrary at-risk process, $A(t)$ a continuous cumulative transition intensity function on $[0,\tau]$, and $h(t)$ a fixed and non-negative function with $h(t)\leq V$ almost everywhere with respect to the Lebesgue-Stieltjes measure generated by $A(t)$. Then, the class of functions
\[
\mathcal{F}_2(s)=\left\{\int_{s}^t h(u)Y(u) dA(u) : t\in[0,\tau]\right\}
\]
is $P$-Donsker for any $s\in[0,\tau)$.
\end{lemma}
\begin{proof}
It is not hard to show that for any probability measure $Q$ and any $t_1,t_2\in[0,\tau]$
\[
\left\|\int_{s}^{t_1}h(u)Y(u)dA(u) - \int_{s}^{t_2}h(u)Y(u)dA(u) \right\|_{Q,2}\leq V|A(t_2)-A(t_1)|.
\]
Now, the class of fixed functions $\Phi_2=\{A(t):t\in[0,\tau]\}$ is a compact subset of $\mathbb{R}$ as it consists of continuous functions on the compact set $[0,\tau]$. Therefore, this class of fixed functions can be covered by $V(1/\epsilon)$ $\epsilon$-balls and, thus, $N(\epsilon,\Phi_2,|\cdot|)\leq V(1/\epsilon)$.
Consequently, for any $t\in[0,\tau]$ there exist a $t_i\in[0,\tau]$, $i=1,\ldots,N(\epsilon,\Phi_2,|\cdot|)$, such that $|A(t)-A(t_i)|<\epsilon$, for any $\epsilon>0$ and any finitely discrete probability measure $Q$. Therefore, for any member of $\mathcal{F}_2(s)$, there exist a $\int_{s}^{t_i}h(u)Y(u)dA(u)$, for $i=1,\ldots,N(\epsilon,\Phi_2,|\cdot|)$, such that
\[
\left\|\int_{s}^{t}h(u)Y(u)dA(u) - \int_{s}^{t_i}h(u)Y(u)dA(u) \right\|_{Q,2}\leq V\epsilon.
\]
for any $\epsilon>0$ and any finitely discrete probability measure $Q$. Consequently, by the minimality of the covering number, it follows that for any $\epsilon>0$ and any finitely discrete probability measure $Q$, we have that
\[
N(\epsilon V,\mathcal{F}_2(s),L_2(Q))\leq V\left(\frac{1}{\epsilon}\right),
\]
which yields a BUEI for $\mathcal{F}_2(s)$. Finally, similar arguments to those used in the proof of Lemma 1 lead to the conclusion that the class $\mathcal{F}_2(s)$ is $P$-Donsker for any $s\in[0,\tau)$.
\end{proof}

\subsection{Regularity conditions}
In this work we assume the following conditions:
\begin{itemize}
\item[C1.] The potential left truncation $L_{im,1}$ and right censoring $L_{im,2}$ times are independent of the underlying counting processes $\{\check{N}_{im,hj}(t):h\neq j, t\in[0,\tau]\}$ and the cluster size $M_i$. Also, $L_{im,1}$ and $L_{im,2}$ are exchangeable in the sense that $EI(L_{im,1}<t,L_{im,2}\geq t)\equiv ER_{im}(t)=ER_{i1}(t)$ for any $i=1,\ldots,n$ and $m=1,\ldots,M_i$.
\item[C2.] The cluster size is bounded in the sense that there exists a (fixed) positive integer $m_0$ such that $\Pr(M>m_0)=0$.
\item[C3.] The underlying counting processes are exchangeable conditionally on cluster size, in the sense that $E[\check{N}_{im,hj}(t)|M_i]=E[\check{N}_{i1,hj}(t)|M_i]$ for any $i=1,\ldots,n$, $m=1,\ldots,M_i$ and $h\neq j$. Also, $E[\check{N}_{im,hj}(\tau)]^2<\infty$ for all $h\neq j$.
\item[C4.] The underlying at-risk processes are exchangeable conditionally on cluster size, in the sense that $E[\check{Y}_{im,h}(t)|M_i]=E[\check{Y}_{i1,h}(t)|M_i]$ for any $i=1,\ldots,n$, $m=1,\ldots,M_i$ and $h\in\mathcal{S}$. Also, there exists a convex and compact set $J_h\subset[0,\tau]$ such that $\inf_{t\in J_h}E[\sum_{m=1}^{M_i}Y_{im,h}(t)]>0$ for all $h\in\mathcal{T}^c$, and $\int_{(0,t]\cap J_h^c}dA_{0,hj}(t)=0$ for all $h\in\mathcal{T}^c$ and $j\neq h$.
\item[C5.] The cumulative transition intensities $\{A_{0,hj}(t):h\neq j, t\in[0,\tau]\}$ are continuous functions.
\item[C6.] Strengthen condition C4 to require $\inf_{t\in [0,\tau]}E[\sum_{m=1}^{M_i}Y_{im,h}(t)]>0$ for all $h\in\mathcal{T}^c$.
\end{itemize}

Conditions C1, C5, and the second parts of conditions C3 and C4 ensure that the standard Aalen--Johansen estimator \citep{Aalen78} of $\mathbf{P}_0$ based on i.i.d. data is uniformly consistent and its elements convergence weakly to tight Gaussian processes.  The additional conditions needed for the situation with clustered data are that cluster sizes are bounded (condition C2), and that the counting and at-risk processes are exchangeable conditionally on cluster size. These additional conditions are realistic in practical applications. Finally, the additional condition C6 is required for the asymptotic linearity of the proposed estimators which provides easy to estimate closed-form variance estimators. In Appendix B we relax condition C6 and show that weak convergence and the validity of the nonparametric cluster bootstrap still hold. In light of the conditional exchangeability of the counting processes (condition C3), condition C2, and the i.i.d. assumption of the observations across clusters it follows that for $h\neq j$ and any $t\in[0,\tau]$
\begin{eqnarray*}
E\left[\sum_{m=1}^{M_1}\check{N}_{1m,hj}(t)\right]&=&E\left\{E\left[\sum_{m=1}^{M_1}\check{N}_{1m,hj}(t)\bigg|M_1\right]\right\} \\
&=& E\left\{E\left[\check{N}_{11,hj}(t)\bigg|M_1\right]\sum_{m=1}^{m_0}I(m\leq M_1)\right\} \\
&=&  E\left\{E\left[\check{N}_{1m,hj}(t)\bigg|M_1\right]M_1\right\} \\
&=&  E\left[M_1\check{N}_{1m,hj}(t)\right], \ \ t\in[0,\tau],
\end{eqnarray*}
for any $m=1,\ldots,M_i$. Similarly, under conditions C2 and C4, it can be shown that $E[\sum_{m=1}^{M_1}\check{Y}_{1m,h}(t)]=E[M_1\check{Y}_{1m,h}(t)]$, $h\in\mathcal{T}^c$, $t\in[0,\tau]$, for any cluster member $m=1,\ldots,M_1$.

For the nonparametric two-sample Kolmogorov--Smirnov tests we refine conditions C3, C4 and C6 as follows:
\begin{itemize}
\item[C3'.] The underlying counting processes are exchangeable conditionally on cluster size, in the sense that $E[\check{N}_{ipm,hj}(t)|M_{pi}]=E[\check{N}_{ip1,hj}(t)|M_{pi}]$ for any $i=1,\ldots,n$, $p=1,2$, $m=1,\ldots,M_{pi}$ and $h\neq j$. Also, $E[\check{N}_{ipm,hj}(\tau)]^2<\infty$ for all $h\neq j$.
\item[C4'.] The underlying at-risk processes are exchangeable conditionally on cluster size, in the sense that $E[\check{Y}_{ipm,h}(t)|M_{pi}]=E[\check{Y}_{ip1,h}(t)|M_{pi}]$ for any $i=1,\ldots,n$, $p=1,2$, $m=1,\ldots,M_{pi}$ and $h\in\mathcal{S}$. Also, there exists a compact set $J_h\subset[0,\tau]$ such that $\inf_{t\in J_h}E[\sum_{m=1}^{M_{pi}}Y_{ipm,h}(t)]>0$, $p=1,2$, for all $h\in\mathcal{T}^c$, and $\int_{J_h^c}dA_{0,phj}(t)=0$, $p=1,2$, for all $h\in\mathcal{T}^c$ and $j\neq h$.
\item[C6'.] Strengthen condition C4' to require $\inf_{t\in [0,\tau]}E[\sum_{m=1}^{M_{pi}}Y_{ipm,h}(t)]>0$, $p=1,2$, for all $h\in\mathcal{T}^c$.
\end{itemize}
Note that the counting and at-risk processes are also allowed to depend on the total cluster size $M_i$. However, the conditional exchangeability assumption in conditions C3' and C4' is defined conditional on the size of the $p$th sample within the $i$th cluster.

\subsection{Proof of Theorem 1}
It is clear that $\check{N}_{im,hj}(t)$, $h\neq j$ can be expressed as
\begin{eqnarray*}
\check{N}_{im,hj}(t)&=&\sum_{v=1}^{\check{N}_{im,hj}(\tau)}I(T_{imv,hj}\leq t) \\
&=&\sum_{v=1}^{v_0}I(v\leq \check{N}_{im,hj}(\tau),T_{imv,hj}\leq t), \ \ a.s.
\end{eqnarray*}
where $T_{imv,hj}$, $v=1,\ldots,\check{N}_{im,hj}(\tau)$, are the random jump times of $\check{N}_{im,hj}(t)$, $t\in[0,\tau]$, and $v_0\in\mathbb{N}$ is a constant which is selected to satisfy $\check{N}_{im,hj}(\tau)\leq v_0$ a.s. in light of condition C3. The corresponding observable version, which is subject to right censoring and/or left truncation, is
\begin{eqnarray*}
N_{im,hj}(t)&=&\sum_{v=1}^{\check{N}_{im,hj}(\tau)}I(T_{imv,hj}\leq t,R_{im}(T_{imv,hj})=1) \\
&=&\sum_{v=1}^{v_0}I(v\leq \check{N}_{im,hj}(\tau),T_{imv,hj}\leq t,R_{im}(T_{imv,hj})=1), \ \ a.s.
\end{eqnarray*}
Thus, by conditions C1 and C2,
\begin{eqnarray*}
EN_{i\cdot,hj}(t)&=&\sum_{m=1}^{m_0}\sum_{v=1}^{v_0}\Pr(m\leq M_i,v\leq \check{N}_{im,hj}(\tau),T_{imv,hj}\leq t,R_{im}(T_{imv,hj})=1) \\
&=&\int_0^tE[R_{i1}(u)]dE\sum_{m=1}^{m_0}\sum_{v=1}^{v_0}I(m\leq M_i,v\leq \check{N}_{im,hj}(\tau),T_{imv,hj}\leq u) \\
&=&\int_0^tE[R_{i1}(u)]dE\check{N}_{i\cdot,hj}(u), \ \ \ \ t\in[0,\tau]
\end{eqnarray*}
Additionally, the observed version of $\check{Y}_{im,h}(t)$, $h\in\mathcal{T}^c$, is $Y_{im,h}(t)=\check{Y}_{im,h}(t)R_{im}(t)$, $t\in[0,\tau]$ and thus, by conditions C1 and C2,
\[
EY_{i\cdot,h}(t)=E[R_{i1}(t)]E[\check{Y}_{i\cdot,h}(t)], \ \ \ \ t\in[0,\tau]
\]
Therefore, using empirical process theory notation and by condition C4 it follows that
\begin{eqnarray*}
\int_0^t\frac{dPN_{\cdot,hj}(u)}{PY_{\cdot,h}(u)}&=&\int_{(0,t]\cap J_h}\frac{dPN_{\cdot,hj}(u)}{PY_{\cdot,h}(u)} \\
&=&\int_{(0,t]\cap J_h}\frac{PR_1(u)dP\check{N}_{\cdot,hj}(u)}{PR_1(u)P\check{Y}_{\cdot,h}(u)} \\
&=&A_{0,hj}(t)
\end{eqnarray*}
since condition C4 ensures $\inf_{t\in[0,t]\cap J_h}PR_{\cdot,h}(t)>0$ and $\int_{(0,t]\cap J_h^c}dA_{0,hj}(t)=0$. Next, it is easy to see that, for any $h\in\mathcal{T}^c$ and $j\in\mathcal{S}$, the following inequality holds:
\begin{eqnarray}
\left\|\hat{A}_{n,hj}(t)-A_{0,hj}(t)\right\|_{\infty}&\leq&\left\|\mathbb{P}_n\int_{(0,t]}\left[\frac{1}{\mathbb{P}_nY_{\cdot,h}(u)}-\frac{1}{PY_{\cdot,h}(u)}\right]dN_{\cdot,hj}(u)\right\|_{\infty} \nonumber \\
&&+\left\|(\mathbb{P}_n-P)\int_{(0,t]}\frac{dN_{\cdot,hj}(u)}{PY_{\cdot,h}(u)}\right\|_{\infty} \nonumber \label{ineq_A}\\
&\equiv& Q_{n,1}+Q_{n,2}.
\end{eqnarray}
The first term can be bounded as follows:
\begin{eqnarray*}
Q_{n,1}&\leq&\|\mathbb{P}_nY_{\cdot,h}(t)-PY_{\cdot,h}(t)\|_{\infty}\left\|\mathbb{P}_n\int_{(0,t]}\frac{dN_{\cdot,hj}(u)}{\mathbb{P}_nY_{\cdot,h}(u)PY_{\cdot,h}(u)}\right\|_{\infty} \\
&\leq&V\|\mathbb{P}_nY_{\cdot,h}(t)-PY_{\cdot,h}(t)\|_{\infty}\left\|\mathbb{P}_n\int_{(0,t]}\frac{dN_{\cdot,hj}(u)}{\mathbb{P}_nY_{\cdot,h}(u)}\right\|_{\infty}
\end{eqnarray*}
where the last inequality follows from condition C4, which implies that there exists a constant $V$ such that $[PY_{\cdot,h}(t)]^{-1}\leq V$ a.e. ($\mu_{N_{\cdot,hj}}$), with $\mu_{N_{\cdot,hj}}$ being the Lebesgue--Stieltjes measure generated by (the sample paths of) $N_{\cdot,hj}(t)$. By conditions C2 and C3, the class of functions $\{Y_{\cdot,h}(t)=\sum_{m=1}^{m_0}I(m\leq M)Y_{m,h}(t):t\in[0,\tau]\}$ can be expressed as a (finite) linear combination of monotone caglad square-integrable processes \citep{Andersen12}, multiplied by $R_{m}(t)$, which belongs to a Donsker class by lemma 4.1. Therefore, by lemma 4.1 and corollary 9.32 in \citet{Kosorok08}, the classes $\{Y_{\cdot,h}(t):t\in[0,\tau]\}$, $h\in\mathcal{T}^c$, are $P$-Donsker and, therefore, also $P$-Glivenko--Cantelli. Consequently, $\|\mathbb{P}_nY_{\cdot,h}(t)-PY_{\cdot,h}(t)\|_{\infty}\overset{as*}\rightarrow 0$. This result and the fact that $[\mathbb{P}_nY_{\cdot,h}(t)]^{-1}$ is bounded a.e. ($\mu_{N_{\cdot,hj}}$) with probability 1 lead to the conclusion that $Q_{n,1}\overset{as*}\rightarrow 0$. For $Q_{n,2}$, conditions C1 and C4 imply that there exists a constant $V$ such that
\[
\frac{1}{PY_{\cdot,h}(t)}\leq V \ \ \ \ a.e. \ \ (\mu_{N_{\cdot,hj}}).
\]
Thus, by conditions C2, C3, and Lemma 1, it follows that the class $\{\int_{(0,t]}[PY_{\cdot,h}(u)]^{-1}dN_{\cdot,hj}(u):t\in[0,\tau]\}$ is $P$-Donsker and thus also $P$-Glivenko--Cantelli. This implies that $Q_{n,2}\overset{as*}\rightarrow 0$ and, consequently, by inequality \eqref{ineq_A} it follows that $\|\hat{A}_{n,hj}(t)-A_{0,hj}(t)\|_{\infty}\overset{as*}\rightarrow 0$, for all $h\in\mathcal{T}^c$ and $j\in\mathcal{S}$. This result along with the continuity of the product integral \citep{Andersen12} lead to the conclusion that
\[
\Prodi_{(0,t]}[\mathbf{I}_k+d\hat{\mathbf{A}}_n(u)]\overset{as*}\rightarrow\Prodi_{(0,t]}[\mathbf{I}_k+d\mathbf{A}_0(u)],
\]
uniformly in $t\in[0,\tau]$.

\subsection{Proof of Theorem 2}
The class of functions $\{N_{\cdot,hj}(t)=\sum_{m=1}^{m_0}I(m\leq M)N_{m,hj}(t):[0,\tau]\}$ is $P$-Donsker for any $h\in\mathcal{T}^c$ and $j\in\mathcal{S}$, by conditions C2 and C3, and lemma 4.1 and corollary 9.32 in \citet{Kosorok08}. Also, the class $\{Y_{\cdot,h}(t):[0,\tau]\}$ is $P$-Donsker for any $h\in\mathcal{T}^c$ as argued in the proof of Theorem 1. Therefore
\[
\sqrt{n}
\begin{pmatrix}
\mathbb{P}_nN_{\cdot,hj}-PN_{\cdot,hj} \\
\mathbb{P}_nY_{\cdot,h}-PY_{\cdot,h}
\end{pmatrix}
\leadsto
\begin{pmatrix}
\tilde{\mathbb{G}}_{1hj} \\
\tilde{\mathbb{G}}_{2h}
\end{pmatrix}
\ \ \ \ \textrm{in} \ \ (D[0,\tau])^2,
\]
for $h\neq j$, where $\tilde{\mathbb{G}}_{1hj}$ and $\tilde{\mathbb{G}}_{2h}$ are tight zero mean Gaussian processes with covariance functions $PN_{\cdot,hj}(t_1)N_{\cdot,hj}(t_2)-PN_{\cdot,hj}(t_1)PN_{\cdot,hj}(t_2)$ and $PY_{\cdot,h}(t_1)Y_{\cdot,h}(t_2)-PY_{\cdot,h}(t_1)PY_{\cdot,h}(t_2)$, respectively, for $t_1,t_2\in[0,\tau]$. The cross-covariance between $\tilde{\mathbb{G}}_{1hj}(t_1)$ and $\tilde{\mathbb{G}}_{2h}(t_2)$ is $PN_{\cdot,hj}(t_1)Y_{\cdot,h}(t_2)-PN_{\cdot,hj}(t_1)PY_{\cdot,h}(t_2)$. Moreover, the map $(F_1,F_2)\mapsto\int_{[0,t]}F_1^{-1}dF_2$ is Hadamard differentiable on the domain
\[
\left\{(F_1,F_2):\inf_{t\in [0,\tau]}|F_1(t)|\geq \epsilon, \int_{[0,\tau]}|dF_2(t)|<\infty\right\}
\]
for $\epsilon>0$ and $F_1^{-1}$ of bounded variation \citep{Kosorok08}, with derivative at $(f_1,f_2)$ given by
\[
\int_{[0,t]} \frac{df_1}{F_2}-\int_{[0,t]}\frac{f_2}{F_2^2}dF_1.
\]
These facts along with condition C6 and the functional delta method \citep{Van00}, lead to the conclusion that
\begin{eqnarray*}
\sqrt{n}(\hat{A}_{n,hj}(t)-A_{0,hj}(t))&=&\sqrt{n}\mathbb{P}_n\Bigg[ \int_{(0,t]}\frac{d[(\mathbb{P}_n-P)N_{\cdot,hj}(u)]}{PY_{\cdot,h}(u)} \\
&&-\int_{(0,t]}\frac{(\mathbb{P}_n-P)Y_{\cdot,h}(u)}{PY_{\cdot,h}(u)}dA_{0,hj}(u)\Bigg]+o_p(1) \\
&=&\sqrt{n}\mathbb{P}_n\left[ \int_{(0,t]}\frac{dN_{\cdot,hj}(u)}{PY_{\cdot,h}(u)}-\int_{(0,t]}\frac{Y_{\cdot,h}(u)}{PY_{\cdot,h}(u)}dA_{0,hj}(u)\right] \\
&&-\sqrt{n}\left[ \int_{(0,t]}\frac{dPN_{\cdot,hj}(u)}{PY_{\cdot,h}(u)}-A_{0,hj}(t)\right]+o_p(1) \\
&=&\sqrt{n}\mathbb{P}_n\left[ \int_{(0,t]}\frac{dN_{\cdot,hj}(u)}{PY_{\cdot,h}(u)}-\int_{(0,t]}\frac{Y_{\cdot,h}(u)}{PY_{\cdot,h}(u)}dA_{0,hj}(u)\right]+o_p(1) \\
&=& \sqrt{n}\mathbb{P}_n\int_{(0,t]}\frac{d\bar{M}_{hj}(u)}{PY_{\cdot,h}(u)}+o_p(1) \\
&\equiv& \sqrt{n}\mathbb{P}_n\phi_{hj}(t)+o_p(1), \ \ \ \ t\in[0,\tau].
\end{eqnarray*}
The class of the influence functions $\{\phi_{hj}(t):t\in[0,
\tau]\}$ is $P$-Donsker by the Donsker property of the class $\{N_{\cdot,hj}(t):[0,\tau]\}$, conditions C2--C5, Lemmas 1 and 2, and corollary 9.32 in \citet{Kosorok08}. Therefore, $\sqrt{n}(\hat{A}_{n,hj}-A_{0,hj})$ converges weakly to a tight zero mean Gaussian process $\tilde{\mathbb{G}}_{3hj}$ in $D[0,\tau]$ with covariance function $P\phi_{hj}(t_1)\phi_{hj}(t_2)$, $t_1,t_2\in[0,\tau]$, for $h\neq j$. For $h=j$, $\sqrt{n}(\hat{A}_{n,hh}(t)-A_{0,hh}(t))=-\sqrt{n}\mathbb{P}_n\sum_{h\neq j}\phi_{hj}(t)+o_p(1)$, where the influence functions belong obviously to a $P$-Donsker class. Thus, the joint sequence $\sqrt{n}(\hat{A}_{n,hj}-A_{0,hj})$ for $h\neq j$, converges weakly to a tight zero mean Gaussian process with cross-covariance between $\tilde{\mathbb{G}}_{3hj}(t_1)$ and $\tilde{\mathbb{G}}_{3lq}(t_2)$ equal to $P\phi_{hj}(t_1)\phi_{lq}(t_2)$, for $h\neq j$, $l\neq q$, $t_1,t_2\in [0,\tau]$. Therefore, $\sqrt{n}(\hat{\mathbf{A}}_n-\mathbf{A}_0)$ converges weakly to a tight zero mean Gaussian process in $(D[0,\tau])^{k^2}$. 
Now, the Hadamard differentiability of the product integral map \citep[proposition II.8.7 in][]{Andersen12}
\[
\mathbf{A}_0\mapsto \Prodi (\mathbf{I}_k-d\mathbf{A}_0),
\]
and the functional delta method \citep{Van00,Andersen12} lead to the conclusion that
\begin{eqnarray*}
\sqrt{n}[\hat{\mathbf{P}}_n(0,t)-\mathbf{P}_0(0,t)]&=&\sqrt{n}\mathbb{P}_n\int_0^{t}\Prodi_{[0,u)}[\mathbf{I}_k+d\mathbf{A}_0(v)]\boldsymbol{\phi}(du)\Prodi_{(u,\cdot]}[\mathbf{I}_k+d\mathbf{A}_0(v)]+o_p(1) \\
&\equiv&\sqrt{n}\mathbb{P}_n\boldsymbol{\gamma}(0,t)+o_p(1), \ \ \ \ t\in[0,\tau]
\end{eqnarray*}
where the matrix $\boldsymbol{\phi}_i(t)$ contains the elements $\phi_{ihj}(t)$, and the matrix $\boldsymbol{\gamma}_i(0,t)$ contains the elements
\[
\gamma_{ihj}(0,t)=\sum_{l\in\mathcal{T}^c}\sum_{q\in\mathcal{S}}\int_0^t\frac{P_{0,hl}(0,u-)P_{0,qj}(u,t)}{PY_{\cdot,l}(u)}d\bar{M}_{ilq}(u), \ \ \ \ t\in[0,\tau].
\]
By the $P$-Donsker property of the classes $\{N_{\cdot,hj}(t):t\in[0,\tau]\}$, for $h\neq t$, and $\{Y_{\cdot,h}(t):t\in[0,\tau]\}$, for $h\in\mathcal{T}^c$, conditions C3-C5, corollary 9.32 in \citet{Kosorok08}, and Lemmas 1 and 2, it follows that the classes $\{\gamma_{hj}(0,t):t\in[0,\tau]\}$ are $P$-Donsker for all $h\in\mathcal{T}^c$, $j\in\mathcal{S}$. This concludes the proof of part (i) of Theorem 2.

For the first conditional weak convergence result in part (ii) of Theorem 2, define the process $\tilde{B}_{hj}(0,t)=\sqrt{n}\mathbb{P}_n\gamma_{hj}(0,t)\xi$. By the $P$-Donsker property of the class $\{\gamma_{hj}(0,t):t\in[0,\tau]\}$ and the conditional multiplier central limit theorem \citep{Kosorok08} it follows that $\tilde{B}_{hj}(0,\cdot)\underset{\xi}{\overset{p}\leadsto}\mathbb{G}_{hj}(0,\cdot)$. Thus, it remains to show that
\[
\|\hat{B}_{hj}(0,t)-\tilde{B}_{hj}(0,t)\|_{\infty}=o_p(1),
\]
unconditionally on the observed data. After some algebra it can be shown that
\begin{equation}
\|\hat{B}_{hj}(0,t)-\tilde{B}_{hj}(0,t)\|_{\infty}\leq \sum_{l\in\mathcal{T}^c}\sum_{q\in\mathcal{S}}(\tilde{Q}_{n,lq1}+\tilde{Q}_{n,lq2}+\tilde{Q}_{n,lq3}), \label{ineqband}
\end{equation}
where
\[
\tilde{Q}_{n,lq1}=\left\|\sqrt{n}\mathbb{P}_n\int_0^t\left[\frac{\hat{P}_{n,hl}(0,u-)\hat{P}_{n,qj}(u,t)}{\mathbb{P}_nY_{\cdot,l}(u)}-\frac{P_{0,hl}(0,u-)P_{0,qj}(u,t)}{PY_{\cdot,l}(u)}\right]dN_{\cdot,lq}(u)\xi\right\|_{\infty},
\]
\[
\tilde{Q}_{n,lq2}=\left\|\sqrt{n}\mathbb{P}_n\int_0^t\left[\frac{\hat{P}_{n,hl}(0,u-)\hat{P}_{n,qj}(u,t)}{\mathbb{P}_nY_{\cdot,l}(u)}-\frac{P_{0,hl}(0,u-)P_{0,qj}(u,t)}{PY_{\cdot,l}(u)}\right]d\hat{A}_{n,lq}(u)\xi\right\|_{\infty},
\]
and
\[
\tilde{Q}_{n,lq3}=\left\|\int_0^t\frac{P_{0,hl}(0,u-)P_{0,qj}(u,t)}{PY_{\cdot,l}(u)}[\sqrt{n}\mathbb{P}_nY_{\cdot,l}(u)\xi]d[\hat{A}_{n,lq}(u)-A_{0,lq}(u)]\right\|_{\infty}.
\]
Next, it is easy to see that
\begin{eqnarray*}
\left|\frac{\hat{P}_{n,hl}(0,u-)\hat{P}_{n,qj}(u,t)}{\mathbb{P}_nY_{\cdot,l}(u)}-\frac{P_{0,hl}(0,u-)P_{0,qj}(u,t)}{PY_{\cdot,l}(u)}\right|&\leq& V\Bigg[\sup_{u\in [0,t]}|\hat{P}_{n,hl}(0,u-)-P_{0,hl}(0,u-)| \\
&&+\sup_{u\in [0,t]}|\hat{P}_{n,hl}(u,t)-P_{0,hl}(u,t)| \\
&& + \sup_{u\in[0,t]}\left|\frac{1}{\mathbb{P}_nY_{\cdot,l}(u)}-\frac{1}{PY_{\cdot,l}(u)}\right|\Bigg],
\end{eqnarray*}
almost everywhere with respect to both $\mu_{N_{\cdot,lq}}$ and $\mu_{\hat{A}_{n,lq}}$ (which is the Lebesgue--Stieltjes measure generated by $\hat{A}_{n,lq}$). Therefore, by condition C3 and C6, the outer almost sure consistency of the transition probability estimators, arguments similar to those used in the proof of Theorem 1, and the central limit theorem, it follows that
\[
\tilde{Q}_{n,lq1}\leq o_{as*}(1)O_p(1)V=o_p(1).
\]
By similar arguments and condition C5 it follows that $\tilde{Q}_{n,lq2}=o_p(1)$. Finally, by the $P$-Donsker property of the class $\{Y_{\cdot,l}(t):t\in[0,\tau]\}$, the uniform consistency of the cumulative transition intensity, and the same arguments to those used in the proof of proposition 7.27 in \citet{Kosorok08}, it follows that $\tilde{Q}_{n,lq3}=o_p(1)$, since convergence in distribution to a constant implies convergence in probability. Thus, by \eqref{ineqband}, $\|\hat{B}_{hj}(0,t)-\tilde{B}_{hj}(0,t)\|_{\infty}=o_p(1)$ and this concludes the proof of the first conditional weak convergence result in part (ii) of Theorem 2.

For the second conditional weak convergence result in part (ii) of Theorem 2, the $P$-Donsker property of the classes $\{N_{\cdot,hj}(t):t\in[0,\tau]\}$ and $\{Y_{\cdot,h}(t):t\in[0,\tau]\}$, condition C3, the weak convergence of the sequence $\sqrt{n}(\hat{A}_{n,hj} - A_{0,hj})$, the bootstrap central limit theorem \citep{Kosorok08}, and the bootstrap functional delta method \citep[Theorem 12.1]{Kosorok08}, imply that $\sqrt{n}(\hat{A}_{n,hj}^* - \hat{A}_{n,hj})\underset{U}{\overset{p}\leadsto}\tilde{\mathbb{G}}_{3hj}$ in $D[0,\tau]$, for $h\in\mathcal{T}^c$ and $j\neq h$. A second application of the bootstrap functional delta method and the bootstrap continuous mapping theorem \citep[Theorem 10.8,][]{Kosorok08} lead to the conclusion that $\sqrt{n}(\hat{P}_{n,hj}^*(0,\cdot) - \hat{P}_{n,hj}(0,\cdot))\underset{U}{\overset{p}\leadsto}\mathbb{G}_{hj}(0,\cdot)$. The proof of part (iii) of Theorem 2 follows from the same arguments.

\subsection{Proof of Theorem 3}
By Theorem 2 and the uniform consistency of $\hat{W}_{hj}(t)$, it follows that
\begin{eqnarray*}
\sqrt{n}\hat{W}_{hj}(t)\hat{\Delta}(0,t) &=& [\hat{W}_{hj}(t)-W_{hj}(t)]\sqrt{n}\mathbb{P}_n[\gamma_{1,hj}(0,t)-\gamma_{2,hj}(0,t)]\\ && +\sqrt{n}\mathbb{P}_nW_{hj}(t)[\gamma_{1,hj}(0,t)-\gamma_{2,hj}(0,t)] + o_p(1) \\
&=&\sqrt{n}\mathbb{P}_nW_{hj}(t)[\gamma_{1,hj}(0,t)-\gamma_{2,hj}(0,t)] + o_p(1).
\end{eqnarray*}
The boundedness of the fixed function $W_{hj}(t)$ and the $P$-Donsker property of $\{\gamma_{p,hj}(0,t):t\in[0,\tau]\}$, $p=1,2$, imply that the class $\{W_{hj}(t)[\gamma_{1,hj}(0,t)-\gamma_{2,hj}(0,t)]:t\in[0,\tau]\}$ is $P$-Donsker. Therefore, $\sqrt{n}\hat{W}_{hj}(\cdot)\hat{\Delta}(0,\cdot)\leadsto\mathbb{Z}_{hj}(0,\cdot)$ in $D[0,\tau]$, with the covariance function of the process $\mathbb{Z}_{hj}(0,\cdot)$ being
\[
W_{hj}(t_1)W_{hj}(t_2)P\{[\gamma_{1,hj}(0,t_1)-\gamma_{2,1hj}(0,t_1)][\gamma_{1,hj}(0,t_2)-\gamma_{2,1hj}(0,t_2)]\},
\]
for $t_1,t_2\in[0,\tau]$.

Next, by the conditional multiplier central limit theorem it follows that
\[
\tilde{C}_{n,hj}(0,\cdot)\equiv\sqrt{P}_nW_{hj}(\cdot)[\gamma_{1,hj}(0,\cdot)-\gamma_{2,hj}(0,\cdot)]\xi\underset{\xi}{\overset{p}\leadsto} \mathbb{Z}_{hj}(0,\cdot) \ \ \ \ in \ \ D[0,\tau].
\]
Also, by the uniform boundedness of $W_{hj}(t)$ and the $P$-Donsker property of the class $\{\gamma_{p,hj}(0,t)\xi:t\in[0,\tau]\}$, it follows that
\begin{eqnarray*}
\sup_{t\in[0,\tau]}\left|\hat{C}_{n,hj}(0,t)-\tilde{C}_{n,hj}(0,t)\right|&\leq& \sum_{p=1}^2\Bigg\{\sup_{t\in[0,\tau]}\Big|[\hat{W}_{hj}(t)-W_{hj}(t)] \\
&&\times \sqrt{n}\mathbb{P}_n[\hat{\gamma}_{p,hj}(0,t)-\gamma_{p,hj}(0,t)]\xi\Big| \\
&& + V\sup_{t\in[0,\tau]}\Big|\sqrt{n}\mathbb{P}_n[\hat{\gamma}_{p,hj}(0,t)-\gamma_{p,hj}(0,t)]\xi\Big| \\
&&+O_p(1)\sup_{t\in[0,\tau]}\left|\hat{W}_{hj}(t)-W_{hj}(t)\right|\Bigg\}.
\end{eqnarray*}
The uniform consistency of $\hat{W}_{hj}(t)$ and the arguments used in the proof of  part (ii) in Theorem 2 lead to the conclusion that $\sup_{t\in[0,\tau]}\left|\hat{C}_{n,hj}(0,t)-\tilde{C}_{n,hj}(0,t)\right|=o_p(1)$ and, thus, $\hat{C}_{n,hj}(0,\cdot)\underset{\xi}{\overset{p}\leadsto} \mathbb{Z}_{hj}(0,\cdot)$ in $D[0,\tau]$.

By Theorem 2 and the bootstrap continuous mapping theorem it follows that
\[
\sqrt{n}W_{hj}(\cdot)[\hat{\Delta}_{n,hj}^*(0,\cdot)-\hat{\Delta}_{n,hj}(0,\cdot)]\underset{U}{\overset{p}\leadsto} \mathbb{Z}_{hj}(0,\cdot) \ \ \ \ in \ \ D[0,\tau].
\]
By the (unconditional) multiplier central limit theorem \citep{Van96} and a double application of the functional delta method, it follows that $\sqrt{n}[\hat{P}_{n,phj}^*(0,\cdot)-\hat{P}_{n,phj}(0,\cdot)]$, $p=1,2$, converge weakly (unconditionally) to tight mean zero Gaussian processes in $D[0,\tau]$. This result along with the uniform consistency of $\hat{W}_{hj}(t)$ lead to the conclusion that
\begin{eqnarray*}
\left\|\sqrt{n}[\hat{W}_{hj}(t)-W_{hj}(t)][\hat{\Delta}_{n,hj}^*(0,t)-\hat{\Delta}_{n,hj}(0,t)]\right\|_{\infty}=o_p(1),
\end{eqnarray*}
unconditionally. Consequently,
\[
\sqrt{n}\hat{W}_{hj}(\cdot)[\hat{\Delta}_{n,hj}^*(0,\cdot)-\hat{\Delta}_{n,hj}(0,\cdot)]\underset{U}{\overset{p}\leadsto} \mathbb{Z}_{hj}(0,\cdot) \ \ \ \ in \ \ D[0,\tau].
\]
Part (ii) of Theorem 3 can be shown using similar arguments.

\section*{Appendix A.5: Violation of condition C6}
\renewcommand{\thesubsection}{A.\arabic{subsection}}

It is possible that, in some applications, condition C6 is not satisfied. This happens when there are transient states with 0 probability of occupation in a subset of the observation time interval $[0,\tau]$. This is the case, for example, in situations where $P_{0,h}(0)=0$ for some transient state(s) $h$. Even though the consistency of the proposed estimators requires only conditions C1-C5, Theorems 2 and 3 additionally require condition C6. If condition C6 is violated for some $h\in\mathcal{T}^c$, and in light of condition C4, it follows that
\[
A_{0,hj}(t)=\int_{(0,t]\cap J_h}\frac{dPN_{\cdot,hj}(u)}{PY_{\cdot,h}(u)},
\]
and
\[
\hat{A}_{n,hj}(t)=\int_{(0,t]\cap J_h}\frac{d\mathbb{P}_nN_{\cdot,hj}(u)}{\mathbb{P}_nY_{\cdot,h}(u)},
\]
where $A_{0,hj}(t)=\hat{A}_{n,hj}(t)=0$ if $t\in[0,t]\cap J_h^c$. In this case, the map $(F_1,F_2)\mapsto\int_{[0,t]\cap J_h}F_1^{-1}dF_2$ is Hadamard differentiable on the domain
\[
\left\{(F_1,F_2):\inf_{t\in J_h}|F_1(t)|\geq \epsilon, \int_{J_h}|dF_2(t)|<\infty\right\}
\]
for $\epsilon>0$ and $F_1^{-1}$ of bounded variation \citep{Kosorok08}. Therefore, the same calculations to those used in the proof of Theorem 2 lead to the conclusion that
\begin{eqnarray*}
\sqrt{n}\{\hat{A}_{n,hj}(t)-A_{0,hj}(t)\}&=&\sqrt{n}\mathbb{P}_n\int_{(0,t]\cap J_h}\frac{d\bar{M}_{hj}(u)}{PY_{\cdot,h}(u)}+o_p(1) \\
&=& \sqrt{n}\mathbb{P}_n\phi_{hj}(t)+o_p(1), \ \ \ \ t\in J_h,
\end{eqnarray*}
with the class $\{\phi_{hj}(t):t\in J_h\}$ being $P$-Donsker. This means that $\sqrt{n}(\hat{A}_{n,hj}-A_{0,hj})$ converges weakly to a tight zero mean Gaussian process $\tilde{\mathbb{G}}_{3hj}$ in $D J_h$ with covariance function $P\phi_{hj}(t_1)\phi_{hj}(t_2)$, $t_1,t_2\in J_h$, for $h\neq j$. The same arguments to those used in the proof of Theorem 2 can be used to show that this theorem holds for $t$ restricted to $\cap_{h\in\mathcal{T}^c}J_h$. This means that inference about $P_{0,hj}(s,t)$, $h\neq j$, is possible for $s$ and $t$ in $\cap_{h\in\mathcal{T}^c}J_h$. From a practical standpoint one needs to restrict the time interval for confidence intervals/bands and hypothesis tests to a set such that there are at least some observations in all transient states.

\end{appendices}

\end{document}